\documentclass[reqno]{amsart}%
	\usepackage[T1]{fontenc}
	\usepackage{lmodern}
	\usepackage[utf8]{inputenc}
	\usepackage[english]{babel}
	\usepackage[english=american]{csquotes}
	\usepackage[pdftex]{graphicx}
	\usepackage{float}
	\usepackage{amssymb,amsmath}
	\usepackage{mathrsfs}
	\usepackage{fancyhdr}
	\usepackage{textcomp}
	\usepackage{color}
	\usepackage{upgreek}
	\usepackage{stackrel}
	\usepackage{enumitem}
	\usepackage[pdftex,plainpages=false,pdfborder={0 0 0},pdfpagelabels,unicode=true,breaklinks=true]{hyperref} 
	
	\usepackage{amsthm}

\theoremstyle{plain}
\newtheorem{Theorem}{Theorem}[section]
\newtheorem{Lemma}[Theorem]{Lemma}

\newtheorem{Corollary}[Theorem]{Corollary}
\newtheorem{Proposition}[Theorem]{Proposition}

\newtheorem*{Proposition*}{Proposition}

\theoremstyle{definition}
\newtheorem{Condition}{Condition}
\newtheorem{Remark}[Theorem]{Remark}
\newtheorem{Example}[Theorem]{Example}
\newtheorem{Definition}[Theorem]{Definition}

\providecommand{\CC}{\mathbb{C}}

\providecommand{\NN}{\mathbb{N}}

\providecommand{\RR}{\mathbb{R}}

\newcommand{\cA}{\mathcal{A}}

\newcommand{\cC}{\mathcal{C}}
\newcommand{\cD}{\mathcal{D}}
\newcommand{\cE}{\mathcal{E}}
\newcommand{\cF}{\mathcal{F}}

\newcommand{\cH}{\mathcal{H}}

\newcommand{\cL}{\mathcal{L}}
\newcommand{\cM}{\mathcal{M}}

\newcommand{\cO}{\mathcal{O}}
\newcommand{\cP}{\mathcal{P}}

\newcommand{\cR}{\mathcal{R}}

\newcommand{\sH}{\mathsf{H}}

\newcommand{\sT}{\mathsf{T}}
\newcommand{\sV}{\mathsf{V}}

\providecommand{\abs}[1]{\left \lvert #1 \right \rvert}

\providecommand{\norm}[1]{\left \lVert #1 \right \rVert}

\providecommand{\scpro}[2]{\left \langle #1 , #2 \right \rangle}

\providecommand{\bscpro}[2]{\bigl \langle #1 , #2 \bigr \rangle}

\DeclareMathOperator\tr{Tr}
\DeclareMathOperator\vol{vol}
\DeclareMathOperator\End{End}

\DeclareMathOperator{\dom}{D}
\DeclareMathOperator\dist{dist}
\DeclareMathOperator{\rank}{rank}
\DeclareMathOperator{\supp}{supp}

\providecommand{\dd}{\mathrm{d}}

\providecommand{\I}{\mathrm{i}}

\providecommand{\iec}{i.e.,~}

\SetMathAlphabet{\mathcal}{normal}{OMS}{cmsy}{m}{n}

\DeclareSymbolFont{lettersA}{U}{txmia}{m}{it}
\DeclareMathSymbol{\updelta}{\mathord}{lettersA}{"0E}

\newcommand{\eps}{\varepsilon}
\newcommand{\Ha}{H_\mathrm{a}}
\newcommand{\Heff}{H_\mathrm{eff}}
\newcommand{\id}{1}

\begin{document}

\title{The Adiabatic Limit of the Connection Laplacian}
\author{Stefan Haag}
\email{stefan.haag@allianz.de}
\address{Allianz Lebensversicherungs-AG, Reinsburgstra{\ss}e 19, 70178 Stuttgart, Germany}

\author{Jonas Lampart}
\email{jonas.lampart@u-bourgogne.fr}
\address{CNRS and Laboratoire Interdisciplinaire Carnot de Bourgogne (UMR 6303 CNRS-Université de Bourgogne Franche-Comté), 9 Avenue Alain Savary, 21078 Dijon, France}

\begin{abstract}
We study the behaviour of Laplace-type operators $H$ on a complex vector bundle $\cE\to M$ in the adiabatic limit of the base space. This space is a fibre bundle $M\to B$ with compact fibres and the limit corresponds to blowing up directions perpendicular to the fibres by a factor $\eps^{-1}\gg 1$. Under a gap condition on the fibre-wise eigenvalues we prove existence of effective operators that provide asymptotics to any order in $\eps$ for $H$ (with Dirichlet boundary conditions), on an appropriate almost-invariant subspace of $L^2(\cE)$.
\end{abstract}

\maketitle


\section{Introduction}
The adiabatic limit of a Riemannian fibre bundle is given by a rescaling, under which the lengths in the fibres (\enquote{vertical directions}) are of order $\eps\ll 1$ compared to the lengths in the base (\enquote{horizontal directions}).
Such scalings arise naturally for some systems studied in physics and also provide a useful tool for geometry, as many properties of the underlying manifold may be analysed more readily in the asymptotic regime $\eps\to 0$.
In this paper we will consider fibre bundles whose fibres are compact manifolds with boundary and whose base is complete, but not necessarily compact. We analyse the asymptotic behaviour of Laplace-type operators on complex vector bundles over these spaces, with Dirichlet conditions on the boundary, in the adiabatic limit.

A particularly interesting example of such an operator is the Hodge Laplacian on (complex valued) differential $p$-forms. In the case of closed manifolds its adiabatic limit was studied by Mazzeo and Melrose~\cite{Mazzeo1990}, who related the asymptotic calculation of its kernel, the $p$-th de Rham cohomology, in the adiabatic limit to Leray's spectral sequence
(see also Forman~\cite{Forman1995}, and \'{A}lvarez L\'{o}pez and Kordyukov~\cite{Lopez2000} for extensions and further references).
Our analysis applies to the Hodge Laplacian (see Example~\ref{ex:Hodge} for details) on rather general non-compact fibre bundles with compact fibres and a boundary.
However, when $\partial M\neq \varnothing$ we treat only Dirichlet boundary conditions and not the releative/absolute conditions related to the respective de Rham cohomolgies on manifolds with boundary.  
Nevertheless, we believe that the techniques we develop can be useful for understanding the $L^2$-cohomology of these manifolds (see Schick~\cite{SchickDiss}) in the adiabatic limit.

In physics, the adiabatic limit appears in the modelling of quantum waveguides, where the manifold in question is a small tubular neighbourhood of a given submanifold of Euclidean space.
The scalar Schrödinger equation in such tubes has been studied thoroughly, an overview of the vast literature can be found in the book of Exner and Kova\v{r}{\'i}k~\cite{Exner2015}. A more geometric view is presented in our recent work with Teufel~\cite{Haag2015}. Similar results were obtained for the heat equation in thin tubes in Riemannian manifolds by Wittich~\cite{Wittich2007}, and Kolb and Krej\v{c}i\v{r}\'{i}k~\cite{Kolb2014}. 
There are also some results for operators built from a non-trivial connection on a line bundle over the waveguide-manifold, which models an external magnetic field~\cite{Exner2001,Borisov2005,Ekholm2005,Bedoya2013,Krejcirik2014,Krejcirik2015}. Our work here will pave the way for the analysis of magnetic fields in the generalised waveguides considered in~\cite{Haag2015}. The details of this application are given in~\cite{Haag2016b}.  Moreover, taking vector bundles of higher rank, our framework allows for the modelling of particles that couple to a non-Abelian gauge field.

Let us now introduce our geometric and analytical setup and give an outline of our approach.
Let $M\xrightarrow{\pi_M} B$ be a connected, smooth fibre bundle of Riemannian manifolds with compact fibre $F$. We assume that the base manifold $B$ is complete, but not necessarily compact, and we allow the total space $M$ to possess a boundary (which is of course the case if and only if $\partial F\neq \emptyset$) and denote by $M_x=\pi_M^{-1}(x)\cong F$ the fibre of $M$ over $x\in B$. 
Let $g$ and $g_B$ be Riemannian metrics on $M$ and $B$, respectively. The tangent bundle of $M$ decomposes into a \textit{vertical} subbundle, the vectors tangent to the fibres, or, equivalently, in the kernel of $\sT \pi_M$, and the \textit{horizontal} vectors, orthogonal to the fibres,
\begin{equation*}
 \sT M = \ker(\sT\pi_M)^\bot \oplus \ker(\sT \pi_M) =: \sH M \oplus \sV M.
\end{equation*}
We assume that $\pi_M$ is a Riemannian submersion for the metrics $g$, $g_B$, that is, $\sT \pi_M{:}\,\sH M \to \sT B$ is an isometry.
The metric $g$ is then called a Riemannian submersion metric and may be written as
\begin{equation*}
g = \pi_M^* g_B + g_\sV,
\end{equation*}
where $g_\sV$ is the restriction of $g$ to the vertical subbundle (and vanishes on horizontal vectors). 
The adiabatic limit is then implemented by a blow-up of the horizontal directions, \iec by the family of rescaled Riemannian submersion metrics
\begin{equation}
\label{eq:geps}
g_\eps = \eps^{-2} \pi_M^* g_B + g_\sV
\end{equation}
for $\eps \ll 1$. Up to a global rescaling, this corresponds to shrinking the length scale of the fibres of $M$ by a factor $\eps$. A simple example of such a family of metrics is given by $g=\eps^{-2}\, \dd x^2 + f(x)\,\dd y^2$ on $\RR\times [0,1]$, for some positive function $f\in C^\infty_b(\RR)$.

Let also $\cE\xrightarrow{\pi_\cE} M$ be a  complex vector bundle with a Hermitian bundle metric $h$ and a metric connection $\nabla^\cE$.
Consider the composition of projections $\Pi_\cE:=\pi_M\circ\pi_\cE$. The fibre of this map $\cE_x = \Pi_\cE^{-1}(x)$ is just the restriction of $\cE$ to the fibre ${M_x}$, and as such itself a vector bundle over $M_x$. In fact, $\Pi_\cE{:}\,\cE\to B$ defines a fibre bundle over $B$, whose typical fibre $\cF\xrightarrow{\pi_\cF} F$ is a vector bundle over $F$. We will examine the construction of this vector bundle, in particular its analytical properties, in Proposition~\ref{prop:bound_cE}.

Let $\cH$ be the Hilbert space of square-integrable sections of $(\cE,h)\xrightarrow{\pi_\cE}(M,g)$ associated with the unscaled Riemannian submersion metric $g=g_{\eps=1}$. We will consider operators of the form
\begin{equation}
\label{eq:HE}
H := - \Delta_{g_\eps}^\cE + \eps H_1 + V
\end{equation}
on $\cH$, with Dirichlet boundary conditions. Here, $-\Delta_{g_\eps}^\cE=-\tr_{g_\eps} (\nabla^\cE)^2=(\nabla^\cE)^*\nabla^\cE$ is the connection Laplacian associated with the connection $\nabla^\cE$ and the metrics $g_\eps$ and $h$, defined by the quadratic form
\[
\scpro{\psi}{-\Delta_{g_\eps}^\cE \psi}_\cH = \int_M \tr_{g_\eps} h\bigl(\nabla^\cE_\cdot\psi,\nabla^\cE_\cdot \psi\bigr)\vol_g.
\]
Moreover, $V$ is an $\End(\cE)$-valued potential and $H_1$ is a perturbation. For example, $H_1$ could be a second-order differential operator modelling small perturbations of the metric $g_\eps$ or the connection $\nabla^\cE$, which appear in typical applications (see Example~\ref{ex:Hodge} and~\cite{Haag2016b}). Our goal will be to obtain precise asymptotics for $H$, its spectrum and dynamics, as $\eps$ tends to zero.

The structure of the rescaled Riemannian submersion~\eqref{eq:geps} yields a splitting of the associated connection Laplacian
\[
-\Delta_{g_\eps}^\cE = -\eps^2 \Delta^\cE_\sH - \Delta_\sV^\cE
\]
into a horizontal Laplacian $\Delta^\cE_\sH = \tr_{\pi_M^* g_B}(\nabla^\cE)^2 - \nabla^\cE_{\eta_\sV}$, where $\eta_\sV$ is the mean curvature vector of the fibres $M_x\hookrightarrow(M,g)$ (for $\eps=1$), and a vertical Laplacian $\Delta_\sV^\cE$. Consequently, the operator~\eqref{eq:HE} takes the form
\[
H = - \eps^2 \Delta^\cE_\sH + \eps H_1 + H^\cF
\]
with fibrewise vertical operator
\[
H^\cF:= -\Delta_\sV^\cE + V.
\]
Under the assumption of bounded geometry (see Section~\ref{sect:bg}), $H$ is self-adjoint and non-negative on the Dirichlet domain $\dom(H)=W^2(\cE)\cap W^1_0(\cE)$ (We denote by $W^k(\cE)=W^{k,2}(\cE)$ the $L^2$-Sobolev space of sections of $\cE$, see Section~\ref{sect:bg}). The same is true for $H^\cF(x)$ acting on sections of $\cE_x\xrightarrow{\pi_{\cE_x}}M_x$, 
with domain $\dom(H^\cF(x))=W^2(\cE_x)\cap W^1_0(\cE_x)$. 
Throughout this paper, we will denote by $H$, $H^\cF$ the unbounded self-adjoint operators defined on their respective domains, while $\Delta_{g_\eps}^\cE$, $\Delta_\sV^\cE$, $\Delta_\sH^\cE$ refer to differential operators, without reference to a specific domain.

Because the fibres $M_x$ are compact, the spectrum of the elliptic operator $H^\cF(x)$ is a discrete set of eigenvalues of finite multiplicity accumulating at infinity. An eigenband is a continuous function $\lambda{:}\,B\to\RR$ that is an eigenvalue of $H^\cF(x)$ for every fixed $x\in B$, \iec $\lambda(x)\in\sigma(H^\cF(x))$ for all $x\in B$. Given an eigenband, we denote by $P_0(x)$ the spectral projection to $\mathrm{ker}\left(H^\cF(x)-\lambda(x)\right)$.
The adiabatic operator associated to the eigenband $\lambda$ is given by
\[
\Ha := P_0 H P_0 = P_0 \bigl(-\eps^2 \Delta^\cE_\sH + \eps H_1\bigr)P_0 + \lambda P_0.
\]
We will show that, under appropriate assumptions, $\Ha$ provides an approximation of $HP_0$ with errors of order $\eps$ and refine this approximation  to accuracy $\eps^{N}$, for arbitrary $N\in \NN$.

Throughout this paper, we will exclusively treat eigenbands with a (local) spectral gap. The generalisation of our results to a group of eigenbands that is separated from the rest of the spectrum is straightforward, but  we will not perform this for the sake of a simpler presentation. The precise condition we require is:

\begin{Condition} \label{cond:gap}
There exist $\delta>0$ and $f_\pm\in C_\mathrm{b}(B)$ with $\dist(f_{\pm}(x),\sigma(H^\cF(x)))\geq \delta$ such that
\[
\bigl[f_-(x),f_+(x)\bigr]\cap \sigma\bigl(H^\cF(x)\bigr) = \lambda(x)
\]
for all $x\in B$. 
\end{Condition}
This condition immediately implies that $\lambda$ is bounded, and also that $\lambda$ is smooth, see Proposition~\ref{prop:R in A}.

It will be convenient to view $H^\cF(x)$, $P_0(x)$ and similar objects as bundle maps on infinite-dimensional vector bundles over $B$ whose fibre at $x\in B$ is given by the space of $L^2$-sections $L^2(\cE_x, h\vert_{\cE_x})$ (or subspaces thereof). 
These vector bundles are constructed as follows: First, note that, by compactness of $M_x\cong F$, the topology of $L^2(\cE_x, h\vert_{\cE_x})$ does not depend on $x$ and the spaces at different points $x$ are isomorphic to a the fixed space $L^2(\cF)$ (as topological vector spaces). A vector bundle with fibre $L^2(\cF)$ is then defined by specifying transition functions $\tau_{U_1, U_2}$ between open sets $U_1,U_2\subset B$. In our case, these are induced by the local trivialisations $\Phi_{j}{:}\,\pi_M^{-1}(U_{j})\to U_{j}\times F$ and $\Psi_{j}{:}\,\Pi_\cE^{-1}(U_{j})\to U_{j}\times\cF$, $j\in \{1,2\}$, of the bundles $M\xrightarrow{\pi_M}B$ and $\cE\xrightarrow{\Pi_\cE}B$, respectively:
 \[
\tau_{U_1, U_2}{:}\,(U_1\cap U_2)\times L^2(\cF)\to (U_1\cap U_2)\times L^2(\cF),\quad (x,\phi) \mapsto \bigl(x,\tau_{U_1, U_2}(x)\phi\bigr)
\]
with
\[
\tau_{U_1, U_2}(x)\phi :=  \Psi_1|_{\cE_x}\circ \Psi_2|^{-1}_{\cE_x}\circ \phi\circ \Phi_2|_{M_x}\circ\Phi_1|_{M_x}^{-1}.
\]
This defines a topological vector bundle over $B$ with typical fibre $L^2(\cF)$ that we denote by $\cH_\cF$. Similarly, we construct a vector bundle $\cD_\cF\subset\cH_\cF$ over $B$ with fibres $(\cD_\cF)_x=\dom(H^\cF(x))$. We treat these as Hermitian vector bundles with the natural pairings induced by $h$, $\nabla^\cE$ and $g_\sV$. The spaces of continuous and fibre-wise maps between vector bundles clearly have a vector bundle structure and the vertical operator $H^\cF$ as well as the associated spectral projection $P_0(x)$ define bounded sections of $\cL(\cD_\cF,\cH_\cF)$  respectively $\cL(\cH_\cF)$.

The finite multiplicity of the eigenvalues of $H^\cF(x)$ immediately implies the finite rank of the projection $P_0(x)$ for all $x\in B$. If $\lambda$ additionally satisfies Condition~\ref{cond:gap}, $P_0$ is a continuous section of $\cL(\cH_\cF)$  (see Proposition~\ref{prop:R in A}), and $\rank(P_0)=\tr(P_0)$ must be constant. Hence, the eigenspace bundle $\cP := P_0 \cH_\cF$ is a well-defined (topological) subbundle of $\cH_\cF$ of finite rank (this bundle also has a natural smooth structure, since $H^\cF$ has a smooth family of eigenfunctions, see~\cite[Prop.~B.7]{Lampart2014}). Via the identification $\cH\cong L^2(\cH_\cF)$ (see \cite[Corollary B.6]{Lampart2014}), the operator $P_0$ defines a bounded operator on $\cH$, whose image $P_0\cH$ is isomorphic to $L^2(\cP)$, the $L^2$-sections of the finite-rank vector bundle $\pi_\cP:\cP\to B$.

The adiabatic operator $\Ha$ acts on $L^2(\cP)$, and on this space we have
\begin{equation*}
 (H-\Ha)P_0= [H, P_0]P_0= \bigl[-\eps^2 \Delta^\cE_\sH + \eps H_1,P_0\bigr]P_0.
\end{equation*}
This commutator is of order $\eps$ as an operator from $\dom(H)$ to $\cH$. We want to caution that for a
 fixed horizontal vector field $X$, $\nabla_{\eps X}^\cE$ is in itself not of order $\eps$, as $\eps X$ is of fixed $g_\eps$-length. However, an expression such as $[\nabla^\cE_{\eps X}, P_0]P_0$, that appears in the commutator, is of order $\eps$ because it is essentially $\eps$-times the derivative of the $\eps$-independent object $P_0$.
 
The commutator $[H,P_0]$ being of order $\eps$ implies that the space $P_0\cH=L^2(\cP)$ is invariant under $H$ up to errors of order $\eps$, and $\Ha$ gives an approximation of $H$ on this space with errors of this order.
Starting from this point, we will improve the approximation and find projections $P_\eps$ such that $P_\eps \cH$ is invariant under $H$ up to errors of order $\eps^N$ (for any given $N\in \NN$).
This generalises the work of the second author with Teufel~\cite{Lampart2016}, where such an approximation was derived for the scalar case $\cE=M\times \CC$, $\nabla^\cE=\dd$.
The generalisation of these results to vector bundles requires an in-depth discussion of the analytical setup. We discuss the structure  of the fibre bundle $\cE \xrightarrow{\Pi_\cE} B$, whose fibre $\cE_x$ is the vector bundle $\cE\vert_{M_x}\cong \cF$, in Section~\ref{sect:bg} and show that it inherits a specific form of bounded geometry from $\cE\to M$ and $M\to B$ in Proposition~\ref{prop:bound_cE}. We then show how $\nabla^\cE$ gives rise to a covariant derivative on $\cH_\cF$, which can be used to calculate objects such as $[\nabla^\cE_{\eps X}, P_0]$ locally over $U\subset B$, even though the fibres of the bundle $\cE \xrightarrow{\Pi_\cE} B$ may themselves have a non-trivial bundle structure $\cF\to F$ (see Lemma~\ref{lem:hor_con}).
With this setup in place, the construction of $P_\eps$ can be performed along the lines of~\cite{Lampart2016}, whose method is inspired by space-adiabatic perturbation theory, which was developed for flat geometries in the context of the Born-Oppenheimer approximation, see~\cite{Martinez2002,Sordoni2003,Teufel2003,Nenciu2004,Panati2007,Martinez2009}. We also provide an improved version of an important lemma in this construction (Lemma~\ref{lem:sadiabatic}).

We will begin by stating our main result and discussing some of its corollaries.
We then work out the details of our geometric and analytical framework
in Section~\ref{sect:bg} and prove the main theorem in Section~\ref{sect:super}.

\subsection{Main Results}

We will assume throughout that the underlying geometry obeys appropriate boundedness properties (Condition~\ref{cond:geometry}, see Section~\ref{sect:bg}), that the potential $V\in C_\mathrm{b}^\infty(\End(\cE))$ is symmetric, smooth and bounded with all its derivatives (cf.~Definition~\ref{def:bound_tensor}) and that the operator $H_1$ satisfies Condition~\ref{cond:H1} (Section~\ref{sect:super}). 
Our main result is the existence of a super-adiabatic projection $P_\eps$, close to $P_0$, which almost commutes with $H$.

\begin{Theorem} \label{thm:Peps}
Let $\lambda$ be an eigenband of $H^\cF$ with a spectral gap (Condition~\ref{cond:gap}) and let $P_0$ be the associated fibre-wise spectral projection.
Then, for all $N\in\NN$ and $\Lambda>0$, there exists an orthogonal projection $P_\eps\in \cL(\cH)\cap \cL(\dom(H))$, satisfying $P_\eps-P_0=\cO(\eps)$ in $\cL(\cH)$ and $\cL(\dom(H))$, such that
\[
\norm{[H,P_\eps]\varrho(H)}_{\cL(\cH)} = \cO(\eps^{N+1})
\]
for every Borel function $\varrho{:}\,\RR\to[0,1]$ with support in $(-\infty,\Lambda]$.
\end{Theorem}
%

Once the construction of the super-adiabatic projection $P_\eps$ is established for some fixed $N\in\NN$ and $\Lambda>0$ we can construct a unitary operator $U_\eps$ that intertwines $P_\eps$ and $P_0$ (\iec $U_\eps P_0 = P_\eps U_\eps$). It is given by the Sz.-Nagy formula (with the abbreviations $P_0^\bot:=\id_\cH - P_0$ and $P_\eps^\bot := \id_\cH - P_\eps$)
\begin{equation*}
U_\eps := \bigl(P_\eps P_0 + P_\eps^\bot P_0^\bot\bigr)\bigl(\id_\cH - (P_0 - P_\eps)^2\bigr)^{1/2}.
\end{equation*}
We then define the associated effective operator 
\begin{equation} \label{eq:Heff}                                                                                                                           
\Heff := U_\eps^* P_\eps H P_\eps U_\eps.                                                                                                                      
\end{equation}

This operator is self-adjoint on $U_\eps^* P_\eps \dom(H)\subset L^2(\cP)$, due to the fact that $[H,P_\eps]=[H,P_0]+\cO(\eps)=\cO(\eps)$ and the Kato-Rellich theorem.
While $H$
acts on $L^2$-sections of $\cE$ (a finite-rank vector bundle over $M$),
$\Heff$ acts on $L^2$-sections of $\cP$ (a finite-rank
vector bundle over the lower dimensional manifold $B$). Hence, the approximation of the initial operator by the effective operator is a dimensional reduction procedure.

The existence of the almost invariant subspace $P_\eps \cH$ for $H$ gives rise to various corollaries on the approximation of spectral and dynamical properties of $H$ using the effective operator $H_\mathrm{eff}$. 
We will state some of these and explain the general ideas behind them. 
These corollaries depend only on the general structure of  Theorem~\ref{thm:Peps} and not on the details of the problem, such as the choice of $\cH=L^2(\cE)$. We thus refer to~\cite{Lampart2016} for complete proofs and focus here on the construction of~$P_\eps$.

The existence of the effective operator allows us to locate (part of) the spectrum of $H$ with high precision, by taking the eigenfunctions (or a Weyl sequence) of $\Heff$ as quasi-modes for $H$.

\begin{Corollary} \label{cor:spect}
Let $N\in \NN$ and $\Lambda>0$ be as in Theorem~\ref{thm:Peps}, and let $\Heff$ be the associated effective operator~\eqref{eq:Heff}. Then, for every $\delta>0$ there exist constants $\eps_0>0$ and $C>0$ such that for all~$\mu\in\sigma(\Heff)$ with $\mu\leq\Lambda-\delta$ one has
\[
\dist\bigl(\mu,\sigma(H)\bigr) \leq C \eps^{N+1}
\]
for all $0<\eps < \eps_0$.
\end{Corollary}

Conversely, we cannot expect to always find spectrum of $\Heff$ near that of $H$. For instance, if the spectrum of $H^\cF$ consists solely of separated bands $\{\lambda_j\}_{j\in\NN}$, then their projections $P_0^j$ give an orthogonal decomposition $\id_\cH = \oplus_{j\in\NN} P_0^j$. By Theorem~\ref{thm:Peps}, $H$ is almost diagonal with respect to this decomposition, so its spectrum is approximated by the union of the spectra of the effective operators. For a given $\mu\in \sigma(H)$ we do not know a priori to which of these sets it is close.

If, however, $\psi\in \dom(H)$ has energy $\scpro{\psi}{H\psi}_\cH\leq \Lambda$ for some $\Lambda\in\RR$ and $-\eps^2 \Delta^\cE_\sH + \eps H_1\geq - C \eps$ is bounded below, 
only finitely many spectral projections $P_0^j$, namely those associated with eigenbands with  $\inf_{x\in B} \lambda_j(x)<\Lambda$, contribute significantly to~$\psi$, because
\[
\Lambda \geq \scpro{\psi}{H\psi}_\cH = \sum_{j\in \NN}\bscpro{\psi}{(-\eps^2\Delta^\cE_\sH + \eps H_1 + \lambda_j)P_0^j \psi}_\cH\geq \inf\lambda_j + \cO(\eps).
\]
If we choose $\Lambda$ small enough, only the ground state band $\lambda_0(x):= \min \sigma\bigl(H^\cF(x)\bigr)$ should contribute and we do expect mutual approximation of the spectra.
In fact, for energies below $\Lambda_1:=\inf_{x\in B}(\sigma(H^\cF)\backslash \lambda_0)$ the operators $H$ and $\Heff$ are almost unitarily equivalent.

\begin{Corollary} \label{cor:equiv}
Let $\lambda_0(x):=\inf \sigma(H^\cF(x))$ be the ground state band. Suppose this satisfies Condition~\ref{cond:gap} and let $\Heff$ be the effective operator~\eqref{eq:Heff} for given constants $N\in \NN$ and $\Lambda>0$.
Assume that $-\eps^2 \Delta^\cE_\sH + \eps H_1$ is bounded from below by $-C\eps$ for some constant $C>0$. 
Then, for every cut-off function  $\chi\in C^\infty_0(\RR)$, with $\chi^p\in C^\infty_0(\RR)$ for all $p\in(0,\infty)$ and support in $(-\infty,\Lambda_1)$, we have
\begin{equation*}
 \norm{U_\eps^* H\chi(H) U_\eps - \Heff\chi(\Heff)}_{\cL(\cH)}=\mathcal{O}(\eps^{N+1}).
\end{equation*}
\end{Corollary}

To gain a better intuition for the effective operator, it is useful to study its expansion in powers of $\eps$.
Using that $U_\eps=P_0+\mathcal{O}(\eps)$, the effective operator
\begin{equation}\label{eq:Heff_exp}
\Heff = U_\eps^* P_\eps H P_\eps U_\eps = P_0 H P_0 + H_\mathrm{sa}
\end{equation}
can be expanded as an adiabatic operator $\Ha$ and a remainder $H_\mathrm{sa}$, which incorporates the super-adiabatic corrections. It follows from the properties of $U_\eps$ that the remainder is of order $\eps^2$ in $\cL(W^2(\cP),L^2(\cP))$.
A formal expansion of $H_\mathrm{sa}$ can be obtained from the explicit construction of $P_\eps$ (see Lemma~\ref{lem:expPeps}). The leading part is given by

\[
\cM := P_0 [P_0,H] R^\cF(\lambda) [P_0,H] P_0,
\]
where $R^\cF(\lambda):=(H^\cF-\lambda)^{-1} P_0^\bot$ stands for the reduced resolvent associated with the eigenband $\lambda$. 
If $H_1$ is a second-order differential operator, the action of $\cM$ may involve horizontal derivatives of fourth order, and then does not define a bounded operator from $W^2(\cP)$ to $L^2(\cP)$. To make the expansion rigorous, we thus need to regularise this expression.
The precise statement is that
\[
\norm{\Heff\chi^2(\Heff) - \chi(\Heff)\bigl(\Ha + \cM\bigr)\chi(\Heff)}_{\cL(L^2(\cP))} = \cO(\eps^3).
\]
for an appropriate cut-off function $\chi\in C^\infty_0$ with support in $(-\infty,\Lambda]$ (see~\cite[Prop.~4.10]{Haag2016a}).

An important example of a Laplace-type operator on a non-trivial vector bundle is the Hodge Laplacian on differential forms. We now explain how to treat its adiabatic limit in our setting. This requires writing the Laplacian using a connection, rather than the (co-) differentials $\dd$ and $\delta_\eps$; see~\cite{Lopez2000} for an exposition in that formalism.
\begin{Example}\label{ex:Hodge}
Let $\cE:=\Lambda^p \sT^*M\otimes \CC$ be the bundle of complex valued $p$-forms on~$M$. This comes naturally with a bundle metric $G_\eps=g_\eps^{\otimes p}$ and an ($\eps$-dependent) connection $\nabla^{\mathrm{LC},p}$ induced by the Levi-Cività connection of $g_\eps$. The Hodge Laplacian (w.r.t.~$g_\eps$) on $\cE$ can be written as
\begin{equation*}
 \delta_\eps\, \dd+ \dd\,\delta_\eps=-\tr_{g_\eps}(\nabla^{\mathrm{LC},p})^2 + W,
\end{equation*}
where the action of the potential $W\in C^\infty(\End(\cE))$ is determined by the curvature tensor $\cR^p$ of $\nabla^{\mathrm{LC},p}$ via
\begin{equation}\label{eq:Weitzenbck}
 (W \omega)(X_1, \dots, X_p)=\sum_{j=1}^p \tr_{g_\eps} \big(\cR^p(\cdot, X_j)\omega\big)(X_1,\dots, X_{j-1}, \cdot, X_{j+1}, \dots, X_p),
\end{equation}
for vector fields $X_1, \dots, X_p$ on $M$.
Note that the potential $W$ depends on $\eps$ through the $g_\eps$-trace and the $\eps$-dependence of $\cR^p$, which is associated with $G_\eps$. Recall also that $(\nabla^{\mathrm{LC},p})^2$ is defined by
\begin{equation*}
 (\nabla^{\mathrm{LC},p})^2(X,Y) = \nabla^{\mathrm{LC},p}_X \nabla^{\mathrm{LC},p}_Y - \nabla^{\mathrm{LC},p}_{\nabla^{g_\eps}_X Y}.
\end{equation*}

In order to fit this into our general setting, we want to choose an $\eps$-independent bundle metric $h$. To do this, we rescale horizontal forms by appropriate factors. Let $P_\sV$ and $P_\sH$ denote the orthogonal projections to the vertical and horizontal bundles $\sV M$, $\sH M$, respectively. Define a map 
\begin{equation*}
 \theta{:}\, \sT M \to \sT M, \, \qquad X\mapsto P_\sV X + \eps^{-1} P_\sH X.
\end{equation*}
 Obviously, $\theta$ is an isometry from $(\sT M, g_\eps)$ to $(\sT M, g)$ (with $\eps=1$). Now choose $h=g^{\otimes p}=G_{\eps=1}$ and let $\theta^p$ be the induced isometry $(\cE, G_\eps)\to (\cE, h)$ (note that $\theta^1=(\theta^T)^{-1}$ maps $\pi_M^*\nu$ to $\eps \pi_M^*\nu$ for $\nu\in C^\infty(\sT^*B\otimes \CC)$). Then
$\nabla^{p,\eps}:=\theta^p \nabla^{\mathrm{LC},p} (\theta^p)^{-1}$ is an $\eps$-dependent family of metric connections on $(\cE, h)$.
The action of these connections is determined by the action of $\theta\nabla^{\mathrm{LC}}\theta^{-1}=:\nabla^\eps$ on vector fields. This can be described as follows: Let $Y_1, Y_2$ be vertical vector fields, $X_1, X_2$ vector fields on $B$, and denote their horizontal lifts by $X_1^*, X_2^*$ (these are the unique horizontal vector fields satisfying $\sT \pi_M X_i^*=X_i$).
Let $\nabla^\sV$ be the Levi-Cività connection of the vertical metric, $W_F{:}\, C^\infty(\sH M)\to C^\infty(\End(\sV M))$ be the Weingarten map of the fibres (for $\eps=1$) and let $\Omega{:}\,C^\infty(\sV M)\to C^\infty(\End(\sH M))$ be the integrability tensor of the horizontal distribution defined by
\begin{equation*}
 g\bigl(\Omega(Y_1)X_1^*,X_2^*\bigr)=g\bigl(Y_1, [X_1^*, X_2^*]\bigr).
\end{equation*}
Using the Koszul formula, one calculates
\begin{align*}
\nabla^\eps_{Y_1} Y_2&= \nabla^{\sV}_{Y_1} Y_2  + \eps^2 P_\sH \nabla^{\eps=1}_{Y_1} Y_2 \\
\nabla^\eps_{X_1^*} X_2^*&= (\nabla^{g_B}_{X_1} X_2)^* + \eps P_\sV [X_1^*, X_2^*] \\
\nabla^\eps_{X_1^*} Y_1&= P_\sV[X_1^*, Y_1] + W_F(X_1^*)Y_1 - \tfrac12 \eps \Omega(Y_1)X_1^*\\
\nabla^\eps_{Y_1} X_1^*&= \eps W_F(X_1^*)Y_1 - \tfrac12 \eps^2 \Omega(Y_1)X_1^*.
\end{align*}
Note that this connection is torsion-free only if $\eps=1$ (or if $W_F=\Omega=0$, in which case the connection is independent of $\eps$). From this we see that $\nabla^{p,\eps}$ depends smoothly on $\eps$ at $\eps=0$ and we may set $\nabla^\cE=\nabla^{p,0}$.
We then have
\begin{equation*}
 \theta^p\tr_{g_\eps}(\nabla^{\mathrm{LC},p})^2(\theta^p)^{-1}=
 \tr_{g_\eps}(\nabla^{p,\eps})^2
 =\tr_{g_\eps}(\nabla^\cE)^2+ \eps \widetilde{H},
\end{equation*}
where 
$\eps \widetilde{H}$ collects all $\eps$-dependent terms of $\nabla^{p,\eps}$. 
Furthermore,
\begin{align*}
 \tr_{g_\eps}(\nabla^\cE)^2&=\tr_{g_\eps}\left(\nabla^\cE_{\cdot}\nabla^\cE_{\cdot} - \nabla^\cE_{\nabla^{\mathrm{LC}}_\cdot \cdot}\right)\\
 &=\tr_{g_\sV} \left(\nabla^\cE_{\cdot}\nabla^\cE_{\cdot} - \nabla^\cE_{\nabla^\sV_\cdot \cdot}\right)
 + \eps^2 \tr_{\pi_M^*g_B} \left(\nabla^\cE_{\cdot}\nabla^\cE_{\cdot} - \nabla^\cE_{\nabla^{g_B}_\cdot \cdot}\right)
 -\eps^2 \nabla^\cE_{\eta_\sV}\\
 &=: \Delta_\sV + \eps^2 \Delta_\sH^\cE,
\end{align*}
with the mean curvature vector of the fibres $\eta_\sV = \tr_{g_\sV}(P_\sH \nabla^{\eps=1}_{\cdot} \cdot)$.
We then set
\begin{align*}
H:=\theta^p(\delta_\eps\, \dd+ \dd\,\delta_\eps)(\theta^p)^{-1}
%
=-\eps^2 \Delta_\sH^\cE - \Delta_\sV^\cE + \eps \widetilde{H} + \theta^pW(\theta^p)^{-1},
\end{align*}
which is of the form we suppose in general.

One can additionally choose $V=\theta^p W(\theta^p)^{-1}\vert_{\eps=0}$, which leads to $\theta^pW(\theta^p)^{-1}=V + \eps \widetilde{V}$, and
\begin{equation*}
 H^\cF:= -\Delta_\sV^\cE + V ,\qquad H_1=\widetilde{H}+ \widetilde{V}.
\end{equation*}
With these choices, we have
\begin{equation*}
 H=-\eps^2 \Delta_\sH^\cE + \eps H_1 + H^\cF.
\end{equation*}

One can check that (note that the trace over horizontal directions yields a pre-factor $\eps^2$; see also~\cite{Lopez2000})
the action of the vertical operator $H^\cF$ on $\omega \wedge \pi_M^*\nu$ with $\omega \in C^\infty\left(\Lambda^q (\sV M)^*\otimes \CC\right)$ and $\nu\in C^\infty\left(\Lambda^{p-q}\sT^*B\otimes \CC\right)$ is just that of the Hodge Laplacian on $q$-forms of the fibre (multiplied by the the identity on $\pi_M^*\Lambda^{p-q}\sT^*B$).
When $\partial M=\emptyset=\partial F$, an example of an eigenband is given by $\lambda(x)\equiv 0$. 
This is the ground state band since $H^\cF= (\dd_F + \dd_F^*)^2$ is non-negative.
It satisfies the gap condition under our boundedness assumptions on the geometry of $M\xrightarrow{\pi_M}B$, explained in Section~\ref{sect:bg}.
The range of the associated projection $P_0(x)$ is given by
\begin{equation*}
 \mathrm{ran}\bigl(P_0(x)\bigr)= \bigoplus_{q=0}^p \mathscr{H}^q(M_x,\CC)\otimes \left(\Lambda^{p-q} T_x^*B\right),
\end{equation*}
where $\mathscr{H}^q(M_x,\CC)$ is the $q$-th de Rham cohomology of $M_x$ with values in $\CC$. The eigenspace bundle $\cP$ is thus given by differential forms on the base with values in the cohomology of the fibres.
By Corollary~\ref{cor:spect}, the effective operator $\Heff$ acting on $L^2$-sections of $\cP$ can be used to find small (approximate) eigenvalues of $H$ with arbitrary precision. For compact $M$ these are related to the Leray spectral sequence for the de Rham cohomology of $M$, see~\cite{Mazzeo1990}.

We remark that while this choice of vertical operator is certainly natural it may sometimes be convenient to add some of the terms in $H_1$, in particular those containing vertical derivatives, to $H^\cF$. This will lead to $\eps$-dependent spectral projections and eigenvalues but this dependence can easily be treated perturbatively, see Remark~\ref{rem:H_F pert}.
In the example at hand, adding all the terms in $H_1$ containing vertical derivatives to $H^\cF$ leads to an operator $-\eps^2\Delta_\sH + \eps H_1$ which satisfies the conditions of Corollary~\ref{cor:equiv}. The small eigenvalues of $H$ are thus exactly those of the effective operator constructed in this way.

We emphasise that when $\partial M\neq\emptyset$ we only treat Dirichlet conditions and not the, perhaps more natural, mixed relative/absolute boundary conditions. This limitation stems from the fact that, in order to view $H^\cF(x)$ as a section of $\cL(\cD_\cF,\cH_\cF)$, we need the boundary condition to be fibrewise. This is of course not the case for Neumann conditions when the normal to $\partial M_x$ with respect to $g_\sV$ is different from that for $\partial M$ with respect to $g_\eps$. Note, however, that these normals become equal in the limit $\eps\to0$, so a generalisation of our methods to such boundary conditions seems possible.
\end{Example}

\section{Bounded geometry}\label{sect:bg}

In order to set up our analysis, we will need a good notion of bounded geometry for the involved objects. We will assume here the basic notions for manifolds without boundary and vector bundles over such manifolds. A detailed discussion of these is given by Eichhorn~\cite{eichhorn07}.

We will use the following, coordinate-independent, definition of bounded tensors.

\begin{Definition}\label{def:bound_tensor}
 Let $(\cE,h,\nabla^\cE)\to (M,g)$ be a vector bundle with bundle metric $h$ and compatible connection $\nabla^\cE$. 
 Equip the bundles $\sT^*M^{\otimes j}\otimes \cE$ with the metric $G$ induced by $g$ and $h$ and the connection $\nabla^j$ induced by the Levi-Cività connection on $\sT M$ and $\nabla^\cE$ on $\cE$.
 A smooth section $\sigma \in C^\infty(\cE)$ is $C^k$-bounded if there is a constant $C(k)$ such that for all $j\in\{0,\dots, k\}$
 \begin{equation*}
  \sup_{p\in M} G_j(\nabla^j \sigma_p, \nabla^j \sigma_p) \leq C(k)\,.
 \end{equation*}
 The set of $C^\infty$-bounded sections, denoted by $C^\infty_\mathrm{b}(\cE)$, is the set of sections that are $C^k$-bounded for all $k\in \NN$.
\end{Definition}

Any definition of bounded geometry for a certain class of objects brings with it the existence of specific normal coordinates. For complete manifolds these are geodesic coordinates around appropriately selected points. For a vector bundle over such a manifold, endowed with a bundle metric and compatible connection, the normal coordinates are given by trivialisations over geodesic coordinate charts, obtained by parallel transport of a given frame along the geodesics.
For manifolds with boundary a concept of bounded geometry was introduced by Schick~\cite{Schick01}.

\begin{Definition}\label{def:delta-bg}
 A Riemannian manifold $(M,g)$ with boundary $\partial M$ is a \textit{$\partial$-manifold of bounded geometry} if the following hold:
 \begin{itemize}
\item \textit{Normal collar}: Let $\nu$ be the inward pointing unit normal of $\partial M$. There exists $r_c>0$ such that the map 
\begin{equation*}
b{:}\,\partial M \times [0,r_c) \to M\,, \qquad (p,t)\mapsto \exp_p(t \nu)
\end{equation*}
is a diffeomorphism to its range.
\item \textit{Injectivity radius of the boundary}: The injectivity radius of $\partial M$ with the induced metric is positive, $r_\mathrm{inj}(\partial M, g\vert_{\partial M})>0$.
\item \textit{Injectivity radius in the interior}: There is $r_{\mathrm{int}}>0$ such that for all $r<r_{\mathrm{int}}$ and $p\in M$ with $\dist(p,\partial M) > r_c/3$ the exponential map restricted to $\mathsf{B}_r(0)\subset T_p M$ is a diffeomorphism to its range.  
\item \textit{Curvature bounds}: The curvature tensor of $M$ and the second fundamental form $S$ of $\partial M$ are $C^\infty$-bounded tensors on $M$ and $\partial M$ respectively, in the sense of Definition~\ref{def:bound_tensor}.
\end{itemize}
\end{Definition}

The normal coordinate charts associated with such a manifold are given by geodesic coordinate charts around points in the interior with $\dist(p,\partial M)>r_c/3$ and, near the boundary, by the composition of a geodesic chart in $\partial M$ with the boundary collar map $b$ defined above.

The following definitions provide the notions of vector bundles of bounded geometry over manifolds with boundary and their natural trivialisations.
\begin{Definition}
A vector bundle with metric connection $(\cE,h,\nabla^\cE)\to (M,g)$ over a Riemannian manifold with boundary is of bounded geometry if $(M,g)$ is of $\partial$-bounded geometry and the curvature tensor $\cR^\cE$ associated to $\nabla^\cE$ and all its covariant derivatives are bounded with respect to the metrics induced by $h$ and $g$.   
\end{Definition}

\begin{Definition}\label{def:adm_frame}
 Let $(\cE,h,\nabla^\cE)\to (M,g)$ be a $\CC^n$-vector bundle of bounded geometry, $U\subset M$ open and $\tau{:}\,\cE\vert_U\to U\times \CC^n$ be a local trivialisation. We call $\tau$ admissible if there is a normal coordinate chart $\kappa{:}\,U\to \RR^m$ and $\tau$ coincides with a trivialisation obtained by parallel transport of an $h$-orthonormal frame along radial geodesics, if $\kappa$ is an interior chart, or geodesics in $\partial M$ composed with geodesics orthogonal to the boundary, if $\kappa$ is a boundary chart. 
\end{Definition}

If $\cE\to M$ is of bounded geometry and we choose normal coordinates, the metric $g$ and the connections $\nabla^\cE$ are expressed by smooth and globally bounded functions. Hence, in this case, Definition~\ref{def:bound_tensor} is equivalent to requiring that the section $\sigma$, expressed in these charts, be smooth and bounded, with global bounds.

In our setting the base manifold $B$ has no boundary and we assume that $(B,g_B)$ is of bounded geometry in the usual sense. The typical fibre $F$ of $M$ is compact, so it is of $\partial$-bounded geometry for any smooth metric. For the total space $M$ we will assume a uniformity condition on its structure, \iec its local trivialisations. 

\begin{Definition}\label{def:uni-triv}
 Let  $(M,g)\xrightarrow{\pi_M}(B,g_B)$ be a Riemannian fibre bundle over a manifold of bounded geometry $B$ with typical fibre $F$.
 The structure $(M,g,\pi_M)$ is said to be \emph{uniformly locally trivial} if there exists a metric $g_F$ on $F$ such that for all $r<r_\mathrm{inj}(B,g_B)$ and $x\in B$ there is a local trivialisation
\begin{equation*}
 \Phi{:}\,\bigl(\pi_M^{-1}(\mathsf{B}_r(x)),g\bigr)\to \bigl(\mathsf{B}_r(x)\times F,g_B\times g_F\bigr)
\end{equation*}
 such that
 \begin{align*}
  &\sT\Phi\in C^\infty_\mathrm{b}\left(\sT^*M\vert_{\pi^{-1}_M(\mathsf{B}_r(x))}\otimes \Phi^{*}\sT(\mathsf{B}_r(x)\times F)\right)\\
  &\sT(\Phi^{-1}) \in C^\infty_\mathrm{b}\left(\sT^*(\mathsf{B}_r(x)\times F)\otimes (\Phi^{-1})^*\sT M\vert_{\pi^{-1}_M(\mathsf{B}_r(x))}\right)
 \end{align*}
are bounded tensors, uniformly in $x$.
\end{Definition}
A detailed discussion of this concept can be found in ~\cite{Lampart2014}. In particular, this property implies that $(M,g)$ is of $\partial$-bounded geometry, as defined above.

\begin{Remark}
 A related concept is that of a foliation of bounded geometry, introduced by Sanguiao~\cite{Sanguiao2008} and given a coordinate-free form by \'{A}lvarez Lopez, Kordyukov and Leichtnam~\cite{LoKoLe2014}. 
 When $\partial M=\emptyset= \partial F $, Definition~\ref{def:uni-triv} implies that $(M,g)$ with its foliation into fibres is a foliation of bounded geometry in the sense of~\cite{Sanguiao2008, LoKoLe2014}. The trivialisations $\Phi$ of $M$ provide suitable coordinate systems of product form, whose size can be estimated by the uniform lower bound on the injectivity radius of $M_x$ (with the induced metric) given by~\cite[Lem. A.7]{Lampart2014}.
 However, bounded geometry in the sense of~\cite{Sanguiao2008, LoKoLe2014} does not imply uniform local trivialisability, because the choice of a fixed reference metric $g_F$ on $F$ gives global (upper) bounds on quantities such as the volume of $M_x$, whereas the purely local definition of~\cite{Sanguiao2008, LoKoLe2014} cannot achieve this.
\end{Remark}

With all the necessary definitions at hand, we can now summarise our conditions on the geometry:

\begin{Condition} \label{cond:geometry}
We require that
\begin{enumerate}
\item $(B,g_B)$\label{glo:B2} is a complete, connected manifold of bounded geometry,
\item $F$ is a compact manifold with boundary,
\item $\pi_M{:}\,(M,g)\to (B,g_B)$ is a uniformly locally trivial fibre bundle\label{glo:M2} with fibre~$F$, 
\item $\pi_\mathcal{E}{:}\,(\mathcal{E},h, \nabla^\cE)\to (M,g)$ is a vector bundle of bounded geometry.
\end{enumerate}
\end{Condition}

Together, these hypothesis imply additional boundedness properties for $(M,g_\eps)$ and the bundle $\Pi_\cE {:}\, \cE\to B$ that will play an important role in our analysis. In the following, Proposition~\ref{prop:bg_M} gives $\eps$-uniform bounded geometry of the total space $(M,g_\eps)$ and Proposition~\ref{prop:bound_cE} shows that the induced bundle $\Pi_\cE{:}\,\cE\to B$ inherits a form of bounded geometry.

\begin{Proposition}\label{prop:bg_M}
Under the hypothesis i)--iii) of Condition~\ref{cond:geometry}, $(M,g)$ is a $\partial$-manifold of bounded geometry in the sense of Definition~\ref{def:delta-bg}.
Furthermore, $(M, g_\eps)$ is of $\partial$-bounded geometry for every $0<\eps\leq 1$ and the constants $r_c$, $r_\mathrm{inj}(\partial M)$, $r_{\mathrm{int}}$ and $C(k)$, $k\in \NN$ can be chosen as those of $(M,g)$.
\end{Proposition}
A proof of this statement can be found in~\cite[Prop. A.4, A.9]{Lampart2014}.

\begin{Proposition}\label{prop:bound_cE}
Asssume Condition~\ref{cond:geometry}.
 The map $\Pi_{\cE}= \pi_M\circ \pi_\cE{:}\, \cE\to B$ is a fibre bundle whose typical fibre $\cF$ is a vector bundle over $F$.
 Moreover, there exists $r_\cE\in \left(0, r_\mathrm{inj}(B,g_B)\right)$ such that for every $x\in B$ and $r<r_\cE$ there is a trivialising bundle map
 \begin{equation*}
  \Theta{:}\, \Pi_{\cE}^{-1}(\mathsf{B}_r(x))\to \mathsf{B}_r(x) \times \cF
 \end{equation*}
with the following boundedness property:
Let $\Phi$ be as in Definition~\ref{def:uni-triv}, let 
$
 \iota_x:=\Phi^{-1}\vert_{\{x\}\times F}
$
be an embedding of the fibre and equip $\cF$ with the pulled-back metric and connection $h_\cF=\iota_x^*h_\cE$, $\nabla^\cF=\iota_x^*\nabla^\cE$. For all admissible trivialisations (Definition~\ref{def:adm_frame})
 $\alpha$ of $\cE\vert_{\pi_M^{-1}(\mathsf{B}_r(x))}$ over $U\subset M$ and $\beta$ of $\cF$ over $V\subset F$, the maps
\begin{align*}
 &\beta\circ\Theta\circ\alpha^{-1}\circ(\Phi^{-1},\id)
 \qquad\qquad\text{and} 
 & (\Phi,\id)\circ\alpha\circ\Theta^{-1}\circ\beta^{-1}
\end{align*}
are linear transformations on $\CC^n$, $C^\infty$-bounded on $\left(\Phi(U)\cap (\mathsf{B}_r(x) \times V)\right)$, uniformly in $x$.
\end{Proposition}
\begin{proof}
The statement that $\Pi_\cE{:}\,\cE\to B$ is a fibre bundle is of course just the existence of local trivialisations. We will prove this in detail, as the explicit construction of trivialisations together with the bounds on $M$ yield boundedness. 

Fix for the moment $x_0 \in B$, $r<r_\mathrm{inj}(B,g_B)$ and $\Phi{:}\,\pi_M^{-1}(\mathsf{B}_r(x_0))\to \mathsf{B}_r(x_0) \times F$, as above.
The idea is to take a diffeotopy $\rho_t$, $t\in[0,1]$, between $\Phi^{-1}$ and $\Phi^{-1}(x_0,\cdot)$ such as
\begin{equation*}
 \rho_t(x,y)=\Phi^{-1}\left(\exp^B_{x_0}(t \exp^{-1}_{x_0}(x)), y\right)\,.
\end{equation*}
This yields an isomorphism of the pulled-back vector bundles (see Hatcher~\cite{Hatcher})
\begin{equation*}
 \rho_1^{*}\cE=(\Phi^{-1})^*\cE\cong \cE\vert_{\Phi^{-1}(\mathsf{B}_r(x_0)\times F)}
\end{equation*}
and
\begin{equation*}
 \rho_0^{*}\cE=\left \lbrace (x,y,e)\in \mathsf{B}_r(x_0) \times F \times \cE \big\vert \pi_\cE(e)=\Phi^{-1}(x_0,y)\right\rbrace
 = \mathsf{B}_r(x_0) \times \iota_{x_0}^*\cE_{{x_0}}\,.
\end{equation*}
Composition of this isomorphism with the map $(\Phi, 1){:}\,\cE\to(\Phi^{-1})^*\cE$ gives a trivialisation. The local trivialisations and connectedness of $B$ clearly imply that the possible choices of $\cF=\iota_x^* \cE_x$ are all isomorphic (as smooth vector bundles).

To show boundedness, we will need to be more specific and make some concrete choices in the construction, which follows that of~\cite[Prop. 1.7]{Hatcher}.
Take $x_0$, $\Phi$ and $\rho_t$ as above. We consider the bundle
\begin{equation*}
 \rho^*\cE=\left\{(t,x,y,e)\in [0,1]\times \mathsf{B}_r(x_0) \times F \times \cE \big\vert \pi_\cE(e)=\rho_t(x,y) \right\}
\end{equation*}
over $[0,1]\times \mathsf{B}_r(x_0) \times F$, with $\rho^*\cE\vert_{t=t_0}=\rho_{t_0}^*\cE$.
The idea is that this interpolates between $\rho_1^*\cE$ and $\rho_0^*\cE$, which allows us to patch together local isomorphisms given by trivialisations. 

The uniform trivialisations of $M$ give us equivalence of the distance functions of $\Phi^*(g_B\times g_F)$ and $g$ as well as $\Phi^*(g_F\vert_{\partial F})$ and $g\vert_{\partial M}$ with uniform constants. This allows us to choose $r_\cE>0$ such that there is a finite system of normal coordinate charts $\{(V_\mu, \kappa_\mu) \vert \mu \in\{1,\dots K\}\}$ of $(F,g_F)$ in a way that for every $x_0\in B$ the set $\Phi^{-1}(\mathsf{B}_r(x_0)\times V_\mu)$ is contained in a normal chart of $M$ if $r<r_\cE$.

Choose from now on $r<r_\cE$ and let $\{\chi_\mu\vert \mu \}$ be a smooth partition  of unity on~$F$ subordinate to the cover $\{V_\mu\vert \mu\}$. Consider the functions $\xi_\mu(y):=\sum_{\nu=1}^\mu \chi_\nu(y)$, $\xi_0\equiv 0$, and their graphs $\Xi_\mu:=\{(t,x,y)\vert t=\xi_\mu(y) \}$. Define a vector bundle over $ \mathsf{B}_r(x_0) \times F$ by $E_\mu:=\rho^*\cE\vert_{\Xi_\mu}$. 
Note that the fibre of $E_\mu$ at $(x,y)$ is the fibre of $\cE$ at $\rho_{\xi_\mu(y)}(x,y)$. Since $\xi_0=0$, $\xi_K=1$, this interpolates between $E_0=\rho_0^*\cE$ and $E_K=\rho_1^*\cE$.

Now let
\begin{equation*}
 \tau_\mu{:}\,\cE_{\Phi^{-1}(\mathsf{B}_r(x_0)\times V_\mu)}\to \Phi^{-1}(\mathsf{B}_r(x_0)\times V_\mu) \times \CC^n
\end{equation*}
be an admissible trivialisation. 
For any $\lambda\in\{0,\dots,K\}$ this gives a trivialisation of $E_\lambda$ over $V_\mu$ by
\begin{equation}\label{eq:def_bg_vartheta}
 \vartheta_{\lambda,\mu}{:}\,E_\lambda\vert_{V_\mu}\to \mathsf{B}_r(x_0)\times V_\mu \times \CC^n \qquad
 (\xi_\lambda(y), x, y, e)\mapsto (x,y, \mathrm{pr}_2\tau_\mu(e)).
\end{equation}
As $\xi_\mu=\xi_{\mu-1}$ on $F\setminus V_\mu$, we have $E_\mu= E_{\mu-1}$ on $\mathsf{B}_r(x_0) \times (F\setminus V_\mu)$ for $\mu\in \{1,\dots ,K\}$. Hence, we can define an isomorphism from $E_\mu$ to $E_{\mu-1}$ by
\begin{align*}
\theta_\mu:=\left\lbrace
\begin{aligned}
 &\vartheta_{\mu-1,\mu}^{-1}\circ \vartheta_{\mu, \mu} \qquad &&\text{over } \mathsf{B}_r(x_0)\times V_\mu\\
 &\id  &&\text{over } \mathsf{B}_r(x_0) \times (F\setminus V_\mu)\,.
\end{aligned}\right.
\end{align*}
Note that $\vartheta_{\mu-1,\mu}^{-1}\circ \vartheta_{\mu, \mu}$ is the identity wherever $\chi_\mu=0$, in particular near the boundary of $V_\mu$, so there is no discontinuity there. Now the finite composition
\begin{equation*}
 \widetilde\Theta:=\theta_1\circ\cdots \circ \theta_K
\end{equation*}
provides a concrete isomorphism from $\rho_1^*\cE=E_K$ to $\rho_0^*\cE=E_0$. We then take
\begin{equation*}
 \Theta:= \widetilde\Theta\circ (\Phi, 1){:}\,\cE\vert_{\pi_M^{-1}(\mathsf{B}_r(x_0))}\to \mathsf{B}_r(x_0) \times \cF\,.
\end{equation*}
Now take trivialisations $\alpha, \beta$ as in the statement of this proposition and let $\widetilde\alpha$ be the induced trivialisation of $(\Phi^{-1})^*\cE$ over $\Phi(U)$. Suppose that $\Phi(U)\cap (\mathsf{B}_r(x_0)\times V)\neq \emptyset$ (otherwise the statement is trivial) and take a point $(x,y)$ in this set. Let $L=\{\mu\in\{1,\dots,K\} \vert y\in V_\mu\}$ and denote the elements of $L$ by $\lambda_0<\dots<\lambda_{|L|}$. Then, on an open neighbourhood of $(x,y)$, we have $\chi_\mu=0$ for $\mu\notin L$ and hence $\widetilde{\Theta}=\theta_{\lambda_1}\circ\cdots\circ\theta_{\lambda_{|L|}}$. Since also $\vartheta_{\lambda_j-1, \lambda_j}=\vartheta_{\lambda_{j-1}, \lambda_j}$ for $j\in \{0,\dots, |L|\}$ (with $\lambda_0=0$), we find
\begin{align}
\beta\circ\widetilde{\Theta}\circ\widetilde{\alpha}^{-1}\label{eq:bg_Theta}
=(\beta\circ \vartheta_{0,\lambda_1}^{-1})
\circ(\vartheta_{\lambda_1,\lambda_1}\circ \vartheta_{\lambda_1,\lambda_2}^{-1})\circ\cdots
\circ(\vartheta_{\lambda_{|L|},\lambda_{|L|}}\circ\widetilde{\alpha}^{-1}).
\end{align}
The intermediate terms $\vartheta_{\lambda_j,\lambda_j}\circ \vartheta_{\lambda_j,\lambda_{j+1}}^{-1}$ are transition functions of the bundle $E_{\lambda_j}$ over $V_{\lambda_j}\cap V_{\lambda_{j+1}}$. In view of~\eqref{eq:def_bg_vartheta}, they are given by the composition of the transition function $\tau_{\lambda_j}^{-1}\circ\tau_{\lambda_{j+1}}$ with the map $(x,y)\mapsto\rho_{\xi_{\lambda_j}(y)}(x,y)$. They are thus smooth, with bounds determined by those of $\tau$, $\chi$ and $(B,g_B)$.

Since $E_{\lambda_{|L|}}=E_K=(\Phi^{-1})^*\cE$ near $(x,y)$, the term on the right of~\eqref{eq:bg_Theta} is just $(\Phi, \id)\circ \tau_{\lambda_{|L|}}\circ\alpha^{-1}\circ (\Phi^{-1}, \id)$, so it is essentially a transition function between admissible trivialisations of $\cE$.

As for the first term, $\vartheta_{0,\lambda_1}$ is obtained from the restriction of $\tau_{\lambda_1}$ to $\cE_x$, but this does not necessarily give an admissible trivialisation of $\cF\to F$, as it is obtained using normal coordinates on $M$, which do not, in general, restrict to normal coordinates on $M_x$ (for example if $M_x$ is not totally geodesic).
However, at the cost of introducing an additional transition function in both $\cF$ and $\cE$, we may assume that $\beta$ and $(\Phi,1)\circ\tau_{\lambda_1}$ are associated with normal coordinates centred at the same point $(x_0, y)$. Since both trivialisations are obtained by parallel transport, the transition function is then just the holonomy of $\nabla^\cE$ along a closed, piecewise smooth, curve in $M$ (e.g. for an interior chart the composition of a geodesic in $(M_x, g_\sV)$ with a geodesic in $(M,g)$ that starts and ends on $M_x$). This can be bounded in terms of the curvature of $\nabla^\cE$ by writing it  as the solution to a differential equation (see Große and Schneider~\cite[Lem.~5.13]{SchGr03}).
Since these bounds are independent of $x$ and $y$ this proves the claim.
\end{proof}

We remark that choosing a different metric and connection on $\cF$ still gives bounded trivialisations, by compactness of $F$. However, the bounds may then depend on $x$, as can be seen by simply scaling a given metric $h_0$ by an $x$-dependent factor $\gamma(x)>0$, $h\vert_{\cE_x}=\gamma(x)h_0$.

We will use Sobolev spaces of sections of $\cE\to M$ that are adapted to the scaling of $g_\eps$. To define these, fix $r<r_{\cE}$ (cf. Proposition~\ref{prop:bound_cE}) and choose points $\lbrace x_\nu\vert\, \nu \in \NN \rbrace$ such that the geodesic balls $U_\nu:=\mathsf{B}_r(x_\nu)$ cover $B$, with finite and globally bounded local multiplicity.
Let $\lbrace \chi_\nu\vert\, \nu \in \NN\rbrace$ be a subordinate partition of unity and $X_1^\nu, \dots, X_d^\nu$, $d=\dim(B)$, an orthonormal frame of vector fields over $U_\nu$, all of which are $C^\infty$-bounded uniformly in $\nu$ (for the existence of these objects see~\cite{eichhorn07}).
Let $\Phi_\nu, \Theta_\nu$ be trivialisations of $\pi_M^{-1}(U_\nu)$, $\Pi_\cE^{-1}(U_\nu)$ with the properties given in Proposition~\ref{prop:bound_cE}. Finally, let $\iota_\nu{:}\, F\to M_{x_\nu}$ be the inclusion of the fibre and set $\cF_\nu:=\iota_\nu^*\cE$ with the induced Riemannian metric on $F$, bundle metric and connection.
Denote furthermore by $\mathrm{pr}_2$ the projection to $\cF_\nu$, the second factor of $U_\nu\times \cF_\nu$, and define $W_{\Theta_\nu}:=\mathrm{pr}_2 \circ \Theta_\nu\circ (\Phi_\nu^{-1})^*$, that is $(W_{\Theta_\nu} \varphi)(x,y)=\mathrm{pr}_2 \Theta_\nu \varphi(\Phi_\nu^{-1}(x,y))$, for any section $\varphi$ of $\cE\vert_{\pi_M^{-1}(U_\nu)}$ and $(x,y)\in U_\nu\times F$.

Then the norm of $W^k_\eps(\cE)$ is defined by
\begin{align}\label{eq:W^k def}
 &\norm{\psi}_{W^k_\eps(\cE)}^2 \\
 &:=\sum_{\nu \in \NN} \sum_{\substack{\alpha\in \NN^d \\ |\alpha|\leq k}} \int_{U_\nu} 
 \norm{  \eps^{|\alpha|}W_{\Theta_\nu}
 \left(\nabla^\cE_{\Phi^*X_1^\nu}\right)^{\alpha_1} \cdots
 \left(\nabla^\cE_{\Phi^*X_d^\nu}\right)^{\alpha_d} (\pi_M^*\chi_\nu) \psi }^2_{W^{k-|\alpha|}(\cF_\nu)}
 \vol_{g_B}. \notag
\end{align}
The norm of $W^\ell(\cF_\nu)$ is defined in the canonical way, see~\cite{SchickDiss} 
for a detailed discussion of these issues. 

The space $W^k_\eps(\cE)$ is the closure of $C^\infty_0(\cE)$ under this norm, and $W^k_{0,\eps}(\cE)$ is the closure of those sections $\psi\in C^\infty_0(\cE)$ with $\supp(\psi)\subset M\setminus \partial M$.
This norm is equivalent to that defined using the metric $g_\eps$, $h$ and $\nabla^\cE$ in~\cite{SchickDiss}, up to a global factor $\eps^d$ due to the scaling of the volume form on $B$. More precisely,
\begin{equation*}
 C^{-1}\norm{\psi}_{W^k_\eps(\cE)}\leq \eps^d \norm{\psi}_{W^k(\cE, g_\eps,h)}\leq C\norm{\psi}_{W^k_\eps(\cE)},
\end{equation*}
with an $\eps$-independent constant.

The ellipticity of $-\Delta^\cE_{g_\eps}$ with Dirichlet conditions (see~\cite{SchickDiss} and \cite[Prop.~2.15]{Haag2016a}) then gives the estimate
\begin{equation}\label{eq:elliptic}
 C^{-1} \left(\norm{\Delta^\cE_{g_\eps}\psi}_{W^{k}_\eps(\cE)} + \norm{\psi}_\cH \right)  \leq \norm{\psi}_{W^{k+2}_\eps(\cE)} \leq C \left(\norm{\Delta^\cE_{g_\eps}\psi}_{W^{k}_\eps(\cE)} + \norm{\psi}_\cH \right)  
\end{equation}
for all $\psi \in W^1_{0,\eps}(\cE)\cap W^{k+2}_\eps(\cE)$, all $k\in \NN$ and some constant $C=C(k)$ independent of $\eps$.

\section{Construction of the almost-invariant subspace} \label{sect:super}

In this section we will construct the super-adiabatic projection $P_\eps$ whose range $P_\eps\cH$ is almost-invariant under $H$. Throughout, we will assume that Condition~\ref{cond:geometry} is satisfied and that the potential $V\in C^\infty_\mathrm{b}(\End(\cE))$ is symmetric.

The projection is constructed recursively, starting from the spectral projection $P_0$, associated to a gapped eigenband $\lambda$, annihilating the commutators with $H$ order by order. To make these calculations rigorous, we work with an adapted calculus that will be introduced in the next subsection. 
We will focus on those aspects that are new compared to the previous work~\cite{Lampart2016}, \iec where the vector bundle structure of $\cE$ plays a role. In parts of the proofs that are essentially independent of this structure we will focus on the main ideas and refer to the existing literature for technical details. A more detailed exposition of these aspects can be found in~\cite{Lampart2016} and the first author's thesis~\cite{Haag2016a}.

\subsection{The framework}
To make sense of commutators with the unbounded operator $(H, \dom(H))$ we need to make sure that certain operators map into the domain of~$H$. This involves keeping track of Sobolev-regularity and boundary conditions. In order to achieve this over several steps in the recursive construction of $P_\eps$, we will now introduce a calculus adapted to our problem.

Commutators such as $[H, P_0]$ can be expressed using commutators with $\nabla^\cE_{X^*}$, with $X\in C^\infty(\sT B)$ (recall that $X^*$ denotes the unique horizontal vector field on $M$ with $\sT \pi_M X^*=X$).
Viewing $P_0$ as a section of $\cL(\cH_\cF)$ and $\cL(\cD_\cF)$, it would be nice to think of $\nabla^\cE_{X^*}$ as a connection on these bundles, induced by $\nabla^\cE$ via the formula $(\nabla_X^{\cH_\cF}\psi)(y)=\nabla^\cE_{X^*_y}\psi(y)$.
There are, however, technical reasons why this approach has to be modified. First, these infinite-rank bundles do not have a natural differentiable structure, so it is not immediately clear to which sections $\psi$ such a formula should apply. Further, the horizontal lift $X^*$ of $X$ might not be tangent to $\partial M$. In that case, this formula does not give rise to a connection on $\cL(\cD_\cF)$, since $\nabla^\cE_{X^*_y}\psi(y)$ will not, in general, satisfy the Dirichlet condition when $\psi$ does.

For these reasons we will define the induced connection on $\cL(\cH_\cF)$ and $\cL(\cD_\cF)$ only locally, which is sufficient since the operators we deal with, e.g.~$[\nabla^\cE_{X^*}, P_0]$, are local with respect to $B$ (in the sense that $\pi_M\left(\supp ([\nabla^\cE_{X^*}, P_0]\psi)\right)\subset \pi_M\left(\supp(\psi)\right)$ for $\psi\in C^\infty_0(\cE)$).

Using the notation of Section~\ref{sect:bg}, let $x_0\in B$ be an arbitrary point, $U=\mathsf{B}_r(x_0)$ be a normal coordinate neighbourhood with radius $r<r_\cE<r_{\mathrm{inj}}(B, g_B)$, $\Phi$ an admissible trivialisation of $\pi_M^{-1}(U)$, and $\Theta$ an admissible trivialisation of $\Pi_\cE^{-1}(U)$.
Let again $W_\Theta:=\mathrm{pr}_2 \circ \Theta\circ (\Phi^{-1})^*$ and note that
\begin{equation*}
  (W_\Theta^{-1}\psi)\bigl(\Phi^{-1}(x,y)\bigr)=\Theta^{-1}(x,\psi(x,y)),
 \end{equation*}
  for a section $\psi$ of $U\times \cF\to U\times F$. Moreover, $W_\Theta$ and $W_\Theta^{-1}$ extend to bounded operators between $L^2(\cE\vert_{\pi_M^{-1}(U)})$ and $L^2(U\times \cF)$, which we interpret as $L^2(\cF)$-valued $L^2$-functions on $U$, \iec $L^2(U, L^2(\cF))$.

\begin{Lemma}\label{lem:hor_con}
 Assume the notation defined above and let $\psi\in C^1(U, L^2(\cF))$, then 
 \begin{equation*}
\nabla^\Theta_X \psi:= W_\Theta\nabla^\cE_{\Phi^*X} W_\Theta^{-1} \psi
 \end{equation*}
 defines a covariant derivative in the sense that
 \begin{equation*}
  \nabla^\Theta_X \psi=X\psi + A^\Theta(X)\psi
 \end{equation*}
with $A^\Theta(X)\in C^\infty(U,\cL(L^2(\cF)))\cap C^\infty(U,\cL(W^2(\cF)\cap W^1_0(\cF)))$. 

The connection form $A^\Theta$ is given by a smooth section of the bundle $\sT^*B\vert_U \times \End(\cF)$ over $U\times F$.
When $\cF$ is equipped with the bundle metric $h_\cF$ and connection $\nabla^\cF$ as in Proposition~\ref{prop:bound_cE} this section and all its derivatives are bounded by constants independent of $U$.
\end{Lemma}

\begin{proof}
 We will explicitly derive the form of $A^\Theta$, all of the claims then follow directly from this. 
Let $V\subset F$ be a normal coordinate neighbourhood. An admissible trivialisation $\tau{:}\,\cF\vert_{V}\to V\times \CC^n$ gives rise to a trivialisation $\tau\circ \Theta$ of $\cE\vert_{\Phi^{-1}(U\times V)}$. Let $\tilde\tau$ be another such trivialisation over $\tilde V$.  If $V\cap \tilde V\neq \emptyset$, the transition function between the trivialisations $\tau\circ \Theta$ and $\tilde{\tau}\circ \Theta$ equals $S:=\tilde \tau\circ \tau^{-1}$, which is the transition function between the trivialisations of $\cF$.
 Let $A^{\Theta,\tau}\in C^\infty(\sT^*(U\times V))\otimes\CC^{n\times n}$ denote the connection form 
 representing the connection $\nabla^\cE$ in the trivialisation $\tau\circ \Theta$. Then, by the specific form of $S$ and the transformation formula for connection coefficients, we have
 \begin{equation*}
  A^{\Theta,\tau}= S^{-1}A^{\Theta,\tilde\tau} S+ \dd \eta,
 \end{equation*}
where $\eta\in C^\infty(V\cap \tilde V,  \CC^{n\times n})$ depends only on the fibre coordinates. Thus, $\dd \eta$ vanishes on vectors tangent to $U\subset B$ and the restrictions of these (matrix-valued) one-forms on $U\times V \cap \tilde V$ to $\sT B$ yields
\begin{equation*}
 A^{\Theta,\tau}\vert_{\sT B}=S^{-1}A^{\Theta,\tilde\tau} \vert_{\sT B} S.
\end{equation*}
These local (on $F$) expressions can thus be patched together, yielding a unique $\End(\cF)$-valued one form $A^\Theta \in C^\infty (\sT^*B\vert_U \times\End(\cF))$ such that
\begin{equation*}
A^\Theta\vert_{V\cap \tilde{V}}=\tau^{-1}\left(A^{\Theta,\tau}\vert_{\sT B}\right)\tau=\tilde\tau^{-1}\left(A^{\Theta,\tilde\tau}\vert_{\sT B}\right)\tilde\tau.
\end{equation*}
As the connection coefficients in admissible trivialisations are uniformly bounded (see Eichhorn~\cite{eichhorn91}, for the boundary charts see also~\cite[Thm.~5.13]{SchGr03}), this shows that $A^\Theta$ is uniformly bounded with all its derivatives by Proposition~\ref{prop:bound_cE}.
Hence, $A^\Theta(v)$, $v\in \sT_x U$, viewed as an operator on $\cL(L^2(\cF))$, restricts to an operator in $\cL(W^k(\cF))$ for every $k\in \NN$. Furthermore, since it acts pointwise, $A^\Theta(v) W^k_0(\cF)\subset W^k_0(\cF)$, and the Dirichlet condition is preserved by $A^\Theta$.
\end{proof}

Observe that the definition of $\nabla^\Theta$ above, together with $\Phi^*[X,Y]=[\Phi^*X, \Phi^*Y]$, immediately gives the curvature operator as
\begin{equation*}
 \cR^\Theta(X,Y)=W_\Theta \cR^{\cE}(\Phi^*X, \Phi^*Y) W_{\Theta}^{-1},
\end{equation*}
where the right hand side has to be understood as a potential in $C^\infty(\End(\cF))$ acting on $L^2(\cF)$.

Now let us sketch how we will use the connection $\nabla^\Theta$ for the computation of commutators. 
Take $X\in C^\infty_\mathrm{b}(\sT B)$ with $g_B$-length of order one (so that $\eps X$ is bounded with respect to $\eps^{-2}g_B$) and consider the commutator $ [\nabla_{\eps X^*}^\cE, P_0]\vert_{\pi_M^{-1}(U)}$.
 Write $X^*=\Phi^*X + Y$ with a smooth and bounded vertical field $Y\in C^\infty_\mathrm{b}(\sV M)$ and introduce $W_\Theta$ to obtain
\begin{equation*}
 [\nabla_{\eps X^*}^\cE, P_0]\vert_{\pi_M^{-1}(U)}=
 W_\Theta^{-1}[\nabla^\Theta_{\eps X}, W_\Theta P_0 W_\Theta^{-1}]W_\Theta\vert_{\pi_M^{-1}(U)} + [\eps \nabla^\cE_Y, P_0]\vert_{\pi_M^{-1}(U)}.
\end{equation*}
The last term is clearly of order $\eps$ in $\cL(\cH_\cF)$, because $P_0$ is bounded from $\cH_\cF$ to $\cD_\cF$, independently of $\eps$. For the first term, we apply Lemma~\ref{lem:hor_con} and obtain
\begin{equation*}
 W_\Theta^{-1}[\eps X + \eps A^\Theta(X), W_\Theta P_0 W_\Theta^{-1}]W_\Theta\vert_{\pi_M^{-1}(U)}.
\end{equation*}
The term involving $\eps A^\Theta(X)$ is of order $\eps$ in $\cL(\cH_\cF)$, by boundedness of $A^\Theta$. The remaining term is essentially $\eps$-times the derivative of $W_\Theta P_0 W_\Theta^{-1}\in L^\infty(U,\cL(L^2(\cF)))$. 
The differentiability of this operator-valued function follows from the gap condition and the differentiability of $W_\Theta H^\cF W_\Theta^{-1}$, which we prove below,
by rather standard arguments. The choice of $\Phi^*X$ over $X^*$ as the direction of derivation is important: The operator $H^\cF$ can only be differentiated in directions which respect the Dirichlet conditions and $\Phi^*X$ is always tangent to $\partial M$.

\begin{Lemma}\label{lem:H_F diff}
 Assume the notation above. 
 Then 
\begin{equation*}
W_\Theta H^\cF W_\Theta^{-1} \in C^\infty_\mathrm{b}\Big(U,  \cL\left(W^2(\cF)\cap W^1_0(\cF),L^2(\cF)\right)\Big).
\end{equation*}
\end{Lemma}

\begin{proof}
For this proof, abbreviate $D:=W^2(\cF)\cap W^1_0(\cF)$ and fix the norm on this space to be the one induced by the Riemannian metric, bundle metric and connection at $x_0\in U$.

Recall that $H^\cF=-\Delta^\cE_\sV + V$. Since the statement for $V$ is trivial, we need to discuss only the vertical Laplacian.
Let $\lbrace V_\mu \vert\, \mu\in\{1, \dots, K\} \rbrace$ be a covering of~$F$ by normal coordinate neighbourhoods and $\chi_\mu$ a subordinate partition of unity. Then observe that, for a local operator $A\in \cL(D, L^2(\cF))$, we have
\begin{align}
 \norm{A \psi}_{L^2(\cF)}\notag
 &=\Big\Vert A \sum_{\mu=1}^K \chi_\mu \psi\Big\Vert_{L^2(\cF)} \leq \sum_{\mu=1}^K \big\Vert A\vert_{V_\mu}\big\Vert_{\cL(D, L^2(\cF))} \norm{\chi_\mu\psi}_D \\
 &\leq C \norm{\psi}_D \sup_\mu  \norm{A\vert_{V_\mu}}_{\cL(D, L^2(\cF))}.\label{eq:loc norm}
\end{align}
The operator $A\vert_{V_\mu}$ denotes the restriction of $A$ to sections with support in $V_\mu$, so its norm is given by the operator norm from $D\cap W^2_0(\cF\vert_{V_\mu})$ to $L^2(\cF)$. This shows that it is sufficient to prove the claim for the restriction of $W_\Theta H^\cF W_\Theta^{-1}$ to an arbitrary $V_\mu$. Denote this set by $V$.

Let $Y_1,\dots Y_m$, $m=\dim(F)$, be smooth, vertical vector fields over $U\times V$ that form an orthonormal frame of $\sT F\vert_{V}$ with respect to $(\Phi^{-1})^*g_\sV$ (see~\cite[Lem.~3.8]{Lampart2016} for a construction) and let $\nabla$ denote the Levi-Cività connection of this metric. Using these, we can write 
\begin{equation*}
\tau\circ W_\Theta \Delta^\cE_\sV W_\Theta^{-1}\circ \tau^{-1}=\sum_{j=1}^m \left(\nabla^{\Theta,\tau}_{Y_j}\nabla^{\Theta,\tau}_{Y_j} - \nabla^{\Theta,\tau}_{\nabla_{Y_j} Y_j}\right),
\end{equation*}
where $\tau:\cF|_V\to V\times\CC^n$ is an admissible trivialisation of $\cF$ over $V$ and $\nabla^{\Theta,\tau}=\dd+A^{\Theta, \tau}$ in the notation used in the proof of Lemma~\ref{lem:hor_con}. This is clearly a differential operator with smooth coefficients. Its (Lie-) derivative $\mathscr{L}_X$ in the direction $X\in C^\infty_\mathrm{b}(\sT B\vert_{U})$ is 
\begin{equation*}
\sum_{j=1}^m \left[\left(\mathscr{L}_X\nabla^{\Theta,\tau}_{Y_j}\right)\nabla^{\Theta,\tau}_{Y_j} 
 +\nabla^{\Theta,\tau}_{Y_j}\left(\mathscr{L}_X\nabla^{\Theta,\tau}_{Y_j}\right)
 - \mathscr{L}_X\nabla^{\Theta,\tau}_{\nabla_{Y_j} Y_j}\right],
\end{equation*}
and for any vertical vector field $Y$
\begin{align*}
 \mathscr{L}_X\nabla^{\Theta,\tau}_{Y}&
 =\mathscr{L}_X \left(Y+ A^{\Theta,\tau}(Y)\right)=[X,Y] + \left(\mathscr{L}_X A^{\Theta,\tau}\right)(Y) + A^{\Theta,\tau} \bigl([X,Y]\bigr)\\
 &=\nabla^{\Theta,\tau}_{[X,Y]} + \left(\mathscr{L}_X A^{\Theta,\tau}\right)(Y).
\end{align*}
Since $[X,Y]$ is a vertical vector field, the derivative of $W_\Theta \Delta_\sV^\cE W_\Theta^{-1}$ is again a second-order differential operator and its norm in $\cL(D, L^2(\cF))$ is bounded since the trivialisation $\tau$, the connection form $A^{\Theta,\tau}$, and the vector fields $Y_j$ are $C^\infty$-bounded. 
Higher-order derivatives can easily be calculated in the same way.
The derivative is in the sense of the operator norm because $V$ is precompact and the vector fields $Y_j$ extend smoothly to $\overline{V}$, so the difference quotients of coefficients converge uniformly to the derivative. This yields the derivative of the operator $(W_\Theta \Delta^\cE_\sV W_\Theta^{-1}, D)$, since the domain $D$ is fixed.
\end{proof}

In order keep track of differential operators and boundary conditions in iterated commutators, we will introduce suitable algebras of operators. These are essentially differential operators in the \enquote{directions} $\Phi^*X$ with coefficients in the fibrewise operators $L^\infty(\cH_\cF)$, respectively $L^\infty(\cH_\cF,\cD_\cF)$, that depend smoothly on the direction $\Phi^*X$.

 \begin{Definition}\label{def:coeff}
  Let the objects $U, \Phi$, $\Theta$, $W_\Theta$ be as above. 
  The coefficient algebras $\cC_U \subset L^\infty(\cL(\cH_\cF)\vert_{\pi_M^{-1}(U)})$ and $\cC_U^H \subset L^\infty(\cL(\cH_\cF,\cD_\cF)\vert_{\pi_M^{-1}(U)})$ consist of those operators $A$, for which $W_\Theta  A W_\Theta^{-1}$ is a smooth function from $U$ to $\cL(L^2(\cF))$ and $\cL(L^2(\cF),W^2(\cF)\cap W^1_0(\cF))$, respectively.
 \end{Definition}

In terms of commutators and in view of Lemma~\ref{lem:hor_con}, for $A\in \cC_U$ this means that
\begin{equation}\label{eq:comm C}
 \Big[\nabla^\cE_{\Phi^*X_1},\big[ \dots [\nabla^\cE_{\Phi^*X_k}, A] \cdots\big]\Big]\in L^\infty\left(\cL(\cH_\cF)\vert_{\pi_M^{-1}(U)}\right),
\end{equation}
for any $k\in \NN$ and vector fields $X_1,\dots, X_k \in C^\infty_\mathrm{b}(\sT B\vert_U)$, and similarly for $A\in\cC_U^H$.

\begin{Definition}\label{def:Algebra}
Assume the notation above and let additionally $X_1,\dots, X_d$ be a $g_B$-orthonormal frame of uniformly $C^\infty$-bounded vector fields over $U$.

The algebra $\cA$ consists of those continuous linear operators $A\in \cL(W^\infty(\cE), L^2(\cE))$ satisfying
\begin{equation*}
 \pi_M \left(\supp A f\right)\subset \pi_M (\supp f)
\end{equation*}
for all sections $f\in W^\infty(\cE)$, such that
\begin{equation}\label{eq:A loc}
 A\vert_{\pi_M^{-1}(U)}=\sum_{\alpha \in \NN^d} A_\alpha(\eps) \eps^{|\alpha|}
 \left(\nabla^\cE_{\Phi^*X_1}\right)^{\alpha_1} \cdots
 \left(\nabla^\cE_{\Phi^*X_d}\right)^{\alpha_d},
\end{equation}
where the sum is finite, $A_\alpha\in \cC_U$ and there exist constants $C(\alpha,m, \eps)$, independent of $U$, for which
\begin{equation}\label{eq:A coeff}
 \norm{W_\Theta A_\alpha(\eps) W_\Theta^{-1}}_{C^m(U,\cL(L^2(\cF))} \leq C(\alpha, m, \eps).
\end{equation}

The algebra $\cA_H\subset \cA$ is the right ideal such that $A_\alpha(\eps)\in \cC_U^H$ and the inequality~\eqref{eq:A coeff} holds in the norm of $C^m\big(U,\cL(L^2(\cF),W^2(\cF)\cap W^1_0(\cF))\big)$.
\end{Definition}

To see that these are indeed algebras, note that~\eqref{eq:comm C} allows us to commute derivatives past the coefficients (see also Equation~\eqref{eq:A comm} below). We may also reorder derivatives, since
\begin{equation}\label{eq:PhiX_comm}
 [\nabla^\cE_{\Phi^*X_i},\nabla^\cE_{\Phi^*X_j}]=\cR^{\cE}( \Phi^*X_i, \Phi^*X_j) +  \nabla^\cE_{[\Phi^*X_i,\Phi^*X_j]}, 
\end{equation}
and $[\Phi^*X_i,\Phi^*X_j]=\Phi^*[X_i,X_j]=\sum_k \Gamma_{ij}^k \Phi^*X_k$, with smooth and bounded coefficients.
Note that the factor $\eps^{|\alpha|}$ in~\eqref{eq:A loc} is natural, because the vector fields $\eps\Phi^*X_i$  have $g_\eps$-length one. 

In order to control the number of derivatives, as well as the dependence on $\eps$, we introduce a double filtration of $\cA$, $\cA_H$ (from now on, we write $\cA_\bullet$ in statements that hold with and without the subscript $H$).
Let $\cA_\bullet^k$ be those $A\in \cA_\bullet$ for which $A_\alpha=0$ for $|\alpha|>k$ (and any $U$). These are differential operators of order at most $k$ in the directions $\Phi^*X_i$.
Then, let $\cA_\bullet^{k,\ell}$ be the elements of order $\eps^\ell$, in the sense that~\eqref{eq:A coeff} holds (in the appropriate norm) with $C(\alpha, m, \eps)=\cO(\eps^\ell)$, for all $\alpha$ and $m$.
Using Equations~\eqref{eq:comm C} and~\eqref{eq:PhiX_comm}, one easily checks that $\cA_\bullet^{k,\ell}\cA_\bullet^{m,n}\subset\cA_\bullet^{k+m, \ell+n}$. 

In view of~\eqref{eq:W^k def} it is clear that the elements of $\cA^k_\bullet$ are bounded operators from $W^k_\eps(\cE)$ to $L^2(\cE)=W^0_\eps(\cE)$. We will denote this operator norm on $\cA^k$ by $\norm{\cdot}_k$. We then have $\norm{A}_k=\cO(\eps^\ell)$ for $A\in \cA_\bullet^{k,\ell}$.
An important property is that, due to the locality of $A\in \cA_\bullet$ w.r.t.~$B$, we can estimate these norms using local bounds, that is
\begin{equation*}
 \norm{A}_k\leq C \sup_{\nu \in\NN} \norm{A\vert_{U_\nu}}_k,
\end{equation*}
where $U_\nu$ is the covering of $B$ used in~\eqref{eq:W^k def} and $A\vert_{U_\nu}$ is the restriction of $A$ to sections with support in $\pi_M^{-1}(U_\nu)$ and the constant $C$ depends on $k$, but not on $\eps$. 
The proof of this estimate is similar to the one for a finite cover~\eqref{eq:loc norm} and relies on the fact that the $U_\nu$'s have bounded local multiplicity; See~\cite[Rem. 3.6]{Haag2016a} for details.

Concerning the algebra $\cA_H$ we have the following lemma.

\begin{Lemma}
 An element  $A\in \cA_H^{k, \ell}$ defines an operator from $W^{k+2}_\eps(\cE)$ to $\dom(H)\subset W^2_\eps(\cE)$, whose norm is of order $\eps^\ell$.
\end{Lemma}
\begin{proof}
 The fact that $A W^{k+2}_\eps(\cE)\subset W^2_\eps(\cE)$ follows immediately from the formula~\eqref{eq:W^k def}, defining the Sobolev norm, the local expression~\eqref{eq:A loc} and the fact that the local coefficients map $L^2(\cE\vert_{M_x})$ to $\dom(H^\cF(x))\subset W^2(\cE\vert_{M_x})$.
The image is contained in the domain of $H$, because the Dirichlet boundary condition is fibrewise and thus $W^2_\eps(\cE) \cap L^2(\cD_\cF)=\dom(H)$.
\end{proof}

We now show that natural objects associated with $H^\cF$ are elements of $\cA$. 

\begin{Proposition}\label{prop:R in A}
 Let $z\in C^\infty_\mathrm{b}(B, \CC)$ with $\dist\big(z(x), \sigma (H^\cF(x))\big)\geq \delta>0$, then $R^\cF(z):=(H^\cF-z)^{-1}\in \cA_H^{0,0}$. If $\lambda$ is an eigenband with a gap (Condition~\ref{cond:gap}) then $P_0$ is a continuous section of $\cL(\cH_\cF)$, $P_0\in \cA_H^{0,0}$ and $\lambda\in C^\infty_\mathrm{b}(B,\RR)$.
\end{Proposition}
\begin{proof}
 For the resolvent, we need to show that $W_\Theta R^\cF(z) W_\Theta^{-1} \in \cC_U^H$. This follows immediately from Lemma~\ref{lem:H_F diff} and the identity
 \begin{align*}
 \mathscr{L}_X W_\Theta R^\cF(z) W_\Theta^{-1}&
 =\mathscr{L}_X (W_\Theta H^\cF W_\Theta^{-1} + z)^{-1}\\
 &= -(W_\Theta H^\cF W_\Theta^{-1} + z)^{-1} \left(\mathscr{L}_X W_\Theta H^\cF W_\Theta^{-1}\right) (W_\Theta H^\cF W_\Theta^{-1} + z)^{-1},
 \end{align*}
for $X\in C^\infty(\sT B\vert_{U})$. The statement for the spectral projection is deduced from this using functional calculus. The smoothness of $\lambda$ follows from smoothness of $W_\Theta H^\cF W_\Theta^{-1}$ and $W_\Theta P_0 W_\Theta^{-1}$.
\end{proof}

\begin{Corollary}\label{cor:R lambda}
 Let $R^\cF$, $\lambda$, $P_0$ as in Proposition~\ref{prop:R in A} and denote $P_0^\perp=\id_{\cH_\cF}-P_0$, then 
 \begin{equation*}
  R^\cF(\lambda)P_0^\perp=P_0^\perp R^\cF(\lambda):=\big(P_0^\perp H^\cF P_0^\perp-\lambda\big)^{-1} \in \cA_H^{0,0}.
 \end{equation*}
\end{Corollary}
\begin{proof}
Proposition~\ref{prop:R in A} shows that the operator $H^\cF P_0^\perp$ is smooth in the sense of Lemma~\ref{lem:H_F diff} and $z(x):=\lambda(x)$ is separated from the spectrum of this operator. Thus, by the arguments of Proposition~\ref{prop:R in A}, its resolvent $\big(H^\cF P_0^\perp-\lambda\big)^{-1}=R^\cF(\lambda)P_0^\perp$ is an element of $\cA_H^{0,0}$.
\end{proof}

\begin{Remark}\label{rem:H_F pert}
 If $H^\cF=H^\cF(\eps)$ depends on $\eps$, it follows easily from the proofs above that $P_0, R^\cF\in \cA_H^{0,0}$, as long as Lemma~\ref{lem:H_F diff} holds with $\eps$-uniform bounds on the derivatives. This is the case, for example, if $V$ depends on $\eps$ but is $\eps$-uniformly bounded with all its derivatives.
 
 In this scenario, we can perform the construction of $P_\eps$ without any changes, because it relies only on the calculus for $\cA_\bullet$.
\end{Remark}

In order to obtain a useful calculus, we still need to consider the horizontal Laplacian $\Delta^\cE_\sH = \tr_{\pi_M^* g_B}(\nabla^\cE)^2 - \nabla^\cE_{\eta_\sV}$.

\begin{Lemma}\label{lem:comm A}
 Let $A\in \cA_H^{k,\ell}$ and $B\in \cA^{k',\ell'}_H$. Then
 \begin{equation*}
  [\eps^2 \Delta^\cE_\sH, A]=\cO(\eps^{\ell+1})
 \end{equation*}
in $\cL(W^{k+2}_\eps(\cE), \cH)$, and 
 \begin{equation*}
  [\eps^2 \Delta^\cE_\sH, A]B\in \cA^{k+k'+1,\ell+\ell'+1}.
 \end{equation*}
\end{Lemma}
\begin{proof}
 Let $U$, $\Phi$ and $X_i$ be as in Definition~\ref{def:Algebra}. With this notation, we have
\begin{equation*}
\Delta^\cE_\sH\vert_{\pi_M^{-1}(U)}= \sum_{i=1}^d \nabla^\cE_{X_i^*}\nabla^\cE_{X_i^*} - \nabla^\cE_{(\nabla_{X_i} X_i)^*} -  (\pi_M^*g_B)\left(X_i^*, \eta_\sV\right)\nabla^\cE_{X_i^*}.
\end{equation*}
Define vertical fields $Y_i$ by $X_i^*=\Phi^*X_i +Y_i$ and expand the expression for $\Delta^\cE_\sH$ using this decomposition.
We then have first-order terms, and second-order terms with any combination of $\nabla^\cE_{\Phi^*X_i}$ and $\nabla^\cE_{Y_j}$.
To prove the claim, we will commute all the derivatives in the local expressions for $[\eps^2 \Delta^\cE_\sH, A]$ to the right.

To compute the commutators, let $D\in \cA_\bullet$  and consider the form of $D$ over $U$ given by Definition~\ref{def:Algebra}. We have
\begin{align}
 [\nabla^\cE_{\Phi^*X_i}, D]\vert_{\pi_M^{-1}}
 = \sum_{\alpha\in \NN} \eps^{|\alpha|} \Big(&
 [\nabla^\cE_{\Phi^*X_i}, D_\alpha] \left(\nabla^\cE_{\Phi^*X_1}\right)^{\alpha_1} \cdots
 \left(\nabla^\cE_{\Phi^*X_d}\right)^{\alpha_d}\notag\\
 & + D_\alpha \Big[\nabla^\cE_{\Phi^*X_i}, \left(\nabla^\cE_{\Phi^*X_1}\right)^{\alpha_1} \cdots
 \left(\nabla^\cE_{\Phi^*X_d}\right)^{\alpha_d}\Big]\Big).\label{eq:A comm}
\end{align}
The coefficient $[\nabla^\cE_{\Phi^*X_i}, D_\alpha]$ is in $\cC_U^\bullet$ by Lemma~\ref{lem:hor_con}, and the second line again contains derivatives of order $|\alpha|$ by~\eqref{eq:PhiX_comm}. Such a commutator thus preserves the order of $D$, both in the number of derivatives and $\eps$. This shows the claim for terms like $[\nabla^\cE_{\Phi^*X_i}\nabla^\cE_{\Phi^*X_j}, A]$ and those containing only one derivative in a direction $\Phi^*X_i$, $i\leq d$.

In order to handle terms containing vertical derivatives, 
observe that
\begin{equation*}
 [\nabla^\cE_{\Phi^*X_i},\nabla^\cE_{Y_j}]=\cR^{\cE}( \Phi^*X_i, Y_j) +  \nabla^\cE_{[\Phi^*X_i,Y_j]}, 
\end{equation*}
where $[\Phi^*X_i,Y_j]$ is vertical, because $\sT\pi_M[\Phi^*X_i, Y_j]=[\sT\pi_M \Phi^*X_i,\sT\pi_M Y_j]=0$. 
We thus have, for example,
\begin{align*}
\big[\nabla^\cE_{\Phi^*X_i},\nabla^\cE_{Y_j} A\big]&=
\nabla^\cE_{Y_j}\underbrace{\big[\nabla^\cE_{\Phi^*X_i}, A\big]}_{\in \cA^{k, \ell}_H} + \bigl(\cR^{\cE}( \Phi^*X_i, Y_j) +  \nabla^\cE_{[\Phi^*X_i,Y_j]}\bigr)A.
\end{align*}
Iterating this calculation and using that $\nabla^\cE_{Y_i} A$ is bounded on $W_\eps^k(\cE)$,  by ellipticity of $\cH_\cF$, we find that $\nabla^\cE_{Y_i} A, \nabla_{Y_i}^\cE \nabla_{Y_j}^\cE A \in \cA^{k,\ell}$, and the same for $B$.
Since $\eps \nabla^\cE_{Y_i}=\cO(\eps)$ in $\cL(W^{p+1}_\eps(\cE), W^p_\eps(\cE))$, for any $p\in \NN$, this proves the claim.
\end{proof}

Lemma~\ref{lem:comm A} implies that
\begin{equation*}
 [\eps^2 \Delta^\cE_\sH, P_0]=[\eps^2 \Delta^\cE_\sH, P_0^2]=[\eps^2 \Delta^\cE_\sH, P_0]P_0 + P_0 [\eps^2 \Delta^\cE_\sH, P_0]=\cO(\eps)
\end{equation*}
in $\cL(W^2_\eps(\cE), \cH)$.
In order to treat the full operator $H$ we need to require appropriate conditions on the perturbation $H_1$. We will express these in terms of the algebras $\cA_\bullet$. It is easy to check that they are satisfied if $H_1$ is a second-order differential operator, symmetric on $\dom(H)$, with ($\eps$-uniformly) bounded coefficients.

\begin{Condition}\label{cond:H1}
The operator $H_1$ is bounded uniformly in $\eps$ from $W_\eps^{p+2}(\cE)$ to $W_\eps^p(\cE)$, for all $p\in\NN$, symmetric on $\dom(H)=W_{0,\eps}^1(\cE)\cap W^2_\eps(\cE)$ and satisfies $H_1 A\in \cA^{k+2,\ell}$, for all $A\in \cA_H^{k,\ell}$. 
\end{Condition}

With this assumption, we have
\begin{equation}\label{eq:P_0 comm}
 [H, P_0]=[-\eps^2 \Delta^\cE_\sH, P_0] + \eps H_1 P_0 -   \eps P_0 H_1=\cO(\eps)
\end{equation}
in $\cL(\dom(H), \cH)$, because $H_1 P_0 \in \cA^{2,0}$ and $P_0 H_1 = \cO(1)$, by Condition~\ref{cond:H1}. Note that we should not consider this expression on $W^2_\eps(\cE)$, as $H$ is not defined there.

\subsection{Proof of Theorem~\ref{thm:Peps}}

We will now prove our main theorem. As a preliminary step, we construct in Lemma~\ref{lem:sadiabatic} below a sequence $P^N$ of almost-projections in $\cA_H$ which are $\eps$-close to $P_0$ and commute with $H$ up to errors of order $\eps^{N+1}$. This lemma is an improved version of a similar statement in~\cite{Lampart2014,Lampart2016} with respect to the order of $P^N\in\cA_H$.
The super-adiabatic projection $P_\eps$ is then obtained from the $P^N$ using a construction due to Nenciu~\cite{Nenciu1993}.

\begin{Lemma}\label{lem:sadiabatic}
For every $k\in \NN$ there exists $P_k \in \mathcal{A}_H^{2k,0}$ such that
\begin{equation*}
 P^N=\sum_{k=0}^{N} \eps^k P_k
\end{equation*}
 satisfies
\begin{enumerate}[label=\arabic*)]
\item $(P^N)^2-P^N \in \mathcal{A}_H^{4N,N+1}$,
\item $\big[H, P^N\big]A\in \mathcal{A}^{2N+2+k,N+1+\ell}$ for all $A\in \cA_H^{k,\ell}$, and
\item $\norm{\big[H, P^N\big]}_{2N+2} = \mathcal{O}(\eps^{N+1})$ on $\dom(H)$.
\end{enumerate}
\end{Lemma}

\begin{proof}
 Take $P_0=P^0$ as a starting point. This is an element of $\cA_H^{0,0}$ by Proposition~\ref{prop:R in A}, satisfies \textit{1)} because it is a projection, and \textit{2)}, \textit{3)}  by Lemma~\ref{lem:comm A} and Condition~\ref{cond:H1}, cf.~Equation~\eqref{eq:P_0 comm}.
 
 The object $P_{N+1}$ is defined recursively by splitting it into diagonal and off-diagonal parts w.r.t~$P_0$.
 Assume \textit{1)} and \textit{2)} hold for $P^N$ and define
\begin{equation*}
 P_{N+1}=P_{N+1}^\mathrm{D}+P_{N+1}^\mathrm{O}, 
\end{equation*}
 with
 \begin{equation*}
  \eps^{N+1}P_{N+1}^\mathrm{O}=-P_0^\perp R^\cF(\lambda)[H, P^N]P_0 + P_0[H, P^N]R^\cF(\lambda) P_0^\perp,
 \end{equation*}
and
\begin{align*}
 \eps^{N+1}P_{N+1}^\mathrm{D}=& -P_0 \left( Q_{N+1} - P^N\right) P_0 + P_0^\perp\left( Q_{N+1} - P^N\right)P_0^\perp,\\
 Q_{N+1}=&\sum_{j=0}^{N+1} \sum_{\substack{k,\ell=0\\ k+\ell=j}}^N \eps^j P_k P_\ell.
\end{align*}
The role of the diagonal part $P_{N+1}^\mathrm{D}$ is to make $P_{N+1}$ an almost-projection in the sense of \textit{1)}, while the off-diagonal part will ensure that \textit{2)} and \textit{3)} hold.

With this definition, $P_{N+1}^\mathrm{O}\in \cA^{2N+2,0}_H$, because \textit{2)} holds for $P^N$ and Corollary~\ref{cor:R lambda} (note also that $P_0^\perp \cA^{k, \ell}\subset \cA^{k, \ell} + P_0 \cA^{k, \ell}\subset \cA^{k, \ell}$). For the diagonal part, first observe that  $P_kP_\ell\in \cA^{2j}\subset \cA^{2N+2}_H$, so $P_{N+1}^\mathrm{D}\in \cA_H^{2N+2}$. Since $Q_{N+1}$ is equal to $(P^N)^2=\sum_{k,\ell=0}^N \eps^{k+\ell} P_k P_\ell$ up to terms in $\cA_H^{4N,N+2}$, the coefficients in~\eqref{eq:A loc} of $Q_{N+1}-P^N$  equal those of $((P^N)^2-P^N)$ (and vanish for $\abs{\alpha}>2N+2$), up to terms of order $\eps^{N+2}$, which shows that $P_{N+1}^\mathrm{D}\in \cA_H^{2N+2,0}$ by \textit{1)}.

In order prove Properties \textit{1)}, \textit{2)} and \textit{3)}  for $P^{N+1}$, one writes all the operators as matrices on $P_0\cH\oplus P_0^\perp \cH$ and treats the four entries separately. 
Since the arguments for the other entries of this matrix are rather similar, we will only treat half of the cases. The calculations for the remaining ones are given in~\cite[Lem.~3.17]{Lampart2016}.

For \textit{1)}, we first consider the $P_0^\perp$-$P_0^\perp$ entry. We expand $P^{N+1}=P^{N}+\eps^{N+1} P_{N+1}$ and then use $P_0^\perp P^N\in P_0^\perp (P_0 + \mathcal{A}_H^{2N,1})\subset \mathcal{A}_H^{2N,1}$, leading to the following calculation:
\begin{align}
P_0^\perp&\big((P^{N+1})^2-P^{N+1}\big) P_0^\perp\notag\\
&= P_0^\perp\big((P^{N}+\eps^{N+1}P_{N+1})^2-P^{N}-\eps^{N+1}P_{N+1}\big)P_0^\perp\notag\\
&\in\begin{aligned}[t]&P_0^\perp\big((P^N)^2-P^N +\eps^{N+1}\big(P^N P_{N+1} + P_{N+1}P^N - P_{N+1}\big)\big)P_0^\perp\\
&+\mathcal{A}_H^{2(2N+2),2N+2}
\end{aligned}\notag\\
&\subset\underbrace{P_0^\perp\big((P^N)^2-P^N\big) P_0^\perp - \eps^{N+1}P_0^\perp P_{N+1}^\mathrm{D} P_0^\perp}_{\in \cA^{4N,N+2}} +\mathcal{A}_H^{4(N+1),N+2}\notag\\
&\subset \mathcal{A}_H^{4(N+1),N+2}.\notag
\end{align}
Next, consider the $P_0$-$P_0^\perp$ entry:

\begin{align}
P_0&\big((P^{N+1})^2 - P^{N+1}\big) P_0^\perp\notag\\
&\in\begin{aligned}[t]& P_0\big((P^N)^2-P^N\big) P_0^\perp +
\eps^{N+1}\overbrace{P_0\big(P^N P_{N+1}-P_{N+1}\big)P_0^\perp}^{\in\mathcal{A}_H^{4N+1,1}}\notag\\
&+ 
\eps^{N+1}\underbrace{P_0  P_{N+1} P^NP_0^\perp}_{\in\mathcal{A}_H^{4N+1,1}}  +\mathcal{A}_H^{2(2N+2),2N+2}
\end{aligned}\\
&\subset P_0\big((P^N)^2-P^N\big) P_0^\perp + \mathcal{A}_H^{4(N+1),N+2}.\notag
\end{align}
Now we use that $Q_{N+1}-P_N \in \cA_H^{2N+2,N+1}$, by the induction hypothesis for $(P^N)^2-P_N$, to obtain
\begin{align*}
 P_0\big((P^N)^2-P_N\big)P_0^\perp&\in P_0(Q_{N+1}-P_N)P_0^\perp +  \cA_H^{4N,N+2}\\
 &\subset P^N(Q_{N+1}-P_N)P_0^\perp + \cA_H^{4N+2,N+2}\\
 &\subset \sum_{\substack{k,\ell,m=0 \\ m+k+\ell \leq N+2}}^N \eps^{k+\ell+m} P_k P_\ell P_m P_0^\perp- (P^N)^2 P_0^\perp+ \cA_H^{4N+2,N+2}\\
 &\subset(Q_{N+1}-P_N)P^N P_0^\perp + \cA_H^{4N+2,N+2}\\
 &\subset \cA_H^{4N+2,N+2},
\end{align*}
which proves the claim.

For \textit{2)}, we let, for simplicity, $A\in \cA_H^{0,0}$ and start with the $P_0$--$P_0^\perp$ entry. By Lemma~\ref{lem:comm A} and Condition~\ref{cond:H1}, we have
\begin{equation*}
 [-\eps^2\Delta^\cE_\sH+\lambda+\eps H_1, P_{N+1}]P_0^\perp A \in \cA^{2N+4,1}.
\end{equation*}
We then obtain 
\begin{equation*}
[H,P_{N+1}]P_0^\perp A\in
[H^\cF-\lambda,P_{N+1}]P_0^\perp A + \cA^{2N+4,1}.
\end{equation*}
This gives us, using that $P_0$ and $R^\cF$ commute and $(H^\cF-\lambda)P_0=0$,
\begin{align}
P_0&\bigl[H, P^N + \eps^{N+1}P_{N+1}\bigr]P_0^\perp A\notag\\
&\in P_0\big([H, P^N] + \eps^{N+1}\big[H^\cF-\lambda,P_{N+1}\big]
\big)P_0^\perp A + \cA^{2N+4,N+2}\notag\\
&\subset P_0\big([H, P^N] + \eps^{N+1}[H^\cF-\lambda, P_{N+1} P_0^\perp]\big) P_0^\perp A +\cA^{2N+4,N+2}\notag\\
&\subset P_0\big([H, P^N] - P_0 P_{N+1} P_0^\perp (H^\cF-\lambda) P_0^\perp A + \cA^{2N+4,N+2}\notag\\
&\subset \underbrace{P_0\bigl([H, P^N] - \big[H,P^N\big]R^\cF(\lambda)P_0^\perp(H^\cF-\lambda)\bigr)P_0^\perp}_{=0} A + \cA^{2N+4,N+2},\notag
\end{align}
which gives \textit{2)} for the $P_0$-$P_0^\perp$-block.

For the $P_0^\perp$-$P_0^\perp$ part, write $P^{N+1}=P^N+\eps^{N+1}P_{N+1}^\mathrm{D} + \eps^{N+1}P_{N+1}^\mathrm{O}$. Then, observe that 
\begin{equation*}
 P_0^\perp \bigl[H, \eps^{N+1}P_0 P_{N+1}P_0^\perp\bigr] P_0^\perp
 = \eps^{N+1} P_0^\perp H P_0 P_{N+1}P_0^\perp
 =\eps^{N+1}[H, P_0]P_0 P_{N+1}P_0^\perp 
\end{equation*}
is an element of $\cA^{2N+4,N+2}$ and thus
\begin{equation*}
 P_0^\perp \bigl[H, \eps^{N+1}P_{N+1}^\mathrm{O} \bigr] P_0^\perp \in \cA^{2N+4,N+2},
\end{equation*}
by the analogous calculation for $P_0^\perp P_{N+1}P_0$.

For the remaining terms, we calculate
\begin{align*}
P_0^\perp&\bigl[H, P^N + \eps^{N+1}P_{N+1}^\mathrm{D}\bigr]P_0^\perp A\notag\\
&=P_0^\perp\big[H, P^N + P_0^\perp(Q_{N+1}-P^N)P_0^\perp\big]P_0^\perp A\notag\\
&=
\begin{aligned}[t]
&P_0^\perp[H, P_0^\perp](Q_{N+1}-P^N)P_0^\perp A + P_0^\perp (Q_{N+1}-P^N)[H, P_0^\perp]P_0^\perp A\\
&+P_0^\perp [H, Q_{N+1}] P_0^\perp A.
  \end{aligned}
\end{align*}
The terms of the first line are in $\cA^{2N+4}$ because $[H, P_0^\perp] A =-[H,P_0] A\in \cA^{2,1}$ and  $(Q_{N+1}-P^N)\in \cA_H^{2N+2,N+1}$.
For the last line, we have
\begin{align*}
 P_0^\perp [H, Q_{N+1}] P_0^\perp A &=
 P_0^\perp \sum_{\substack{k,\ell=0\\ k+\ell\leq N+1}}^N \eps^{k+\ell} \left( \big[H, P_k \big] P_\ell+ P_k\big[H, P_\ell\big] \right) P_0^\perp A\\
 &=P_0^\perp \left( \sum_{\ell=1}^N  \big[H, P^{N+1-\ell} \big] \eps^{\ell} P_\ell+ \sum_{k=1}^N P_k \big[H, P^{N+1-k}\big] \right) P_0^\perp A\\
 &\in \cA^{2N+4,N+2}\,,
\end{align*}
by using the induction hypothesis on $P_k$, $P_\ell$.
This completes the proof of \textit{2)}, and the reasoning for \textit{3)} is essentially the same.
\end{proof}

\begin{proof}[Proof of Theorem~\ref{thm:Peps}]
To complete the proof, we need to construct, for any given $N$ and $\Lambda$, a projection $P_\eps\in \cL(\cH)\cap \cL(\dom(H))$ such that
\begin{equation*}
 \norm{[H, P_\eps]\varrho(H)}_{\cL(\cH)}=\cO(\eps^{N+1}),
\end{equation*}
for every measurable function $\varrho{:}\,[-\infty, \Lambda]\to [0,1]$.
$P_\eps$ will be obtained from $P^N$ by a construction that goes back to Nenciu~\cite{Nenciu1993} and has been used in many later works. We will thus only sketch this procedure, a complete presentation adapted to our notation can be found in~\cite{Lampart2016}.

The first point is to note that, by Condition~\ref{cond:H1}, the domain of $(H)^k$ is contained in $W^{2k}_\eps(\cE)$, for $\eps$ small enough, and, due to the ellipticity of $-\Delta^\cE_{g_\eps}$ with Dirichlet conditions (Equation~\eqref{eq:elliptic}) its graph norm is equivalent to the one of $W^{2k}_\eps(\cE)$. Thus, by choosing an appropriate cut-off function $\chi$, we can define a regularised version of $P^N$ by
\begin{equation*}
 P^\chi:=P_0 + (P^N-P_0)\chi(H) + \chi(H)(P^N-P_0)(1-\chi(H)).
\end{equation*}
One then checks that $P^\chi\in \cL(\cH)$ is self-adjoint and $P^\chi=P_0+\cO(\eps)$, in both $\cL(\cH)$ and $\cL(\dom(H))$.
Since $P^\chi$ is close to $P_0$, we immediately have $[H, P^\chi]=\cO(\eps)$. If $\chi$ is chosen to equal one on $\supp \varrho$ we also have
\begin{equation*}
 [H, P^\chi]\varrho(H)=[H, P^N]\varrho(H)=\cO(\eps^{N+1}),
\end{equation*}
by Lemma~\ref{lem:sadiabatic}, as $\varrho(H):\cH\to W^{2N+2}_\eps(\cE) \cap \dom(H)$ is bounded. 

The operator $P^\chi$ is not a projection, but, since it is close to $P_0$, its spectrum is contained in $\eps$-balls around zero and one. We then define the projection $P_\eps$ as the spectral projection of $P^\chi$ to the spectrum contained in a ball of radius $1/2$ around one,
\begin{equation*}
 P_\eps=\frac{\I}{2\pi} \int_\gamma (P^\chi - z)^{-1} \dd z,
\end{equation*}
where $\gamma$ is the boundary of this ball.
One easily checks that $P_\eps-P_0=\cO(\eps)$.
The proof is then completed by an analysis of the functional calculus of $H$ and $P^\chi$, which we do not replicate here.
\end{proof}

This completes the proof of our results. Let us remark that the expansion $P^N=\sum_{k=0}^N \eps^k P_k$ of Lemma~\ref{lem:sadiabatic} leads to an asymomptotic expansion of $P_\eps$, following ~\cite{Nenciu1993} or~\cite[Lem.~2.25]{Lampart2014}.

\begin{Lemma} \label{lem:expPeps}
 Let $P_k$, $k\leq N$, be as in Lemma~\ref{lem:sadiabatic} and $P_\eps$ be the associated super-adiabatic projection. Then, 
 \begin{equation*}
  \norm{\bigl(P_\eps - \textstyle{\sum}_{k=0}^\ell\, \eps^k P_k \bigr)\chi(H)}_{\cL(\cH, \dom(H))} = \cO(\eps^{\ell +1}),
 \end{equation*}
for all $\ell\leq N$ and every cut-off $\chi\in C^\infty_0((-\infty, \Lambda], [0,1])$ with $\chi^p\in C^\infty_0$ for all $p\in \NN$.
\end{Lemma}

From this, one obtains the expansion of the effective operator~\eqref{eq:Heff_exp}, see~\cite[Prop.~4.10]{Haag2016a}.

\section*{Acknowledgements}

S.~H. would like to thank his supervisor Stefan Teufel for his great support throughout the whole PhD project. Moreover, he is indepted to the DFG Research Training Group 1838 \textit{Spectral Theory and
Dynamics of Quantum Systems} and the CEREMADE (UMR 7534 CNRS and Université Paris-Dauphine) for supporting his research stays in Paris in 2014 and 2016.
J.~L. thanks Sebastian Goette for discussions.

\providecommand{\bysame}{\leavevmode\hbox to3em{\hrulefill}\thinspace}
\providecommand{\MR}{\relax\ifhmode\unskip\space\fi MR }
\providecommand{\MRhref}[2]{%
  \href{http://www.ams.org/mathscinet-getitem?mr=#1}{#2}
}
\providecommand{\href}[2]{#2}

\end{document}